\newcommand*{\TECHREP}{}%

\ifdefined\PAPER
\documentclass[sigconf]{acmart}
\fi
\ifdefined\TECHREP
\documentclass[letterpaper,twocolumn,10pt]{article}
\pdfoutput=1
\usepackage{usenix}
\fi

\usepackage{epsfig,endnotes}
\usepackage{color, colortbl}
\usepackage{tikz}
\usepackage{graphicx}
\usepackage{amssymb,amsmath}
\usepackage{booktabs}
\usepackage{verbatim}
\usepackage{multirow}
\usepackage{subfig}
\usepackage{hyperref,cleveref}
\usepackage{url}
\usepackage[linesnumbered,ruled,vlined]{algorithm2e}
\usepackage[utf8]{inputenc}
\usepackage{amsthm}

\usepackage{times}

\hypersetup{
    colorlinks=true,
    linkcolor=black,
    citecolor=black,
    filecolor=black,
    urlcolor=black,
}

\ifdefined\TECHREP
\newtheorem{theorem}{Theorem}[section]
\newtheorem{corollary}{Corollary}[theorem]
\newtheorem{lemma}[theorem]{Lemma}
\fi

\crefname{section}{§}{§§}
\Crefname{section}{§}{§§}

\definecolor{DarkGray}{gray}{0.85}

\newcommand{\LSC}[0]{\textsc{AllConcur}}
\newcommand{\LSCextended}[0]{Algorithm for LeaderLess CONCURrent atomic broadcast}
\newcommand{\tobroadcast}[1]{\emph{A}-\emph{broa\-dcast}$(\mathit{#1})$}
\newcommand{\todeliver}[1]{\emph{A}-\emph{deliver}$(\mathit{#1})$}
\newcommand{\rbroadcast}[1]{\emph{R}-\emph{broa\-dcast}$(\mathit{#1})$}
\newcommand{\rdeliver}[1]{\emph{R}-\emph{deliver}$(\mathit{#1})$}
\newcommand{\sender}[1]{$\mathit{sender}(\mathit{#1})$}

\newcommand{\code}[1]{Source code: \url{https://github.com/mpoke/allconcur/commit/c09dee8f8f186ee7b2d4fdb23e682016eb3dbde8}}

\DeclareMathAlphabet{\mathpzc}{OT1}{pzc}{m}{it}

\newlength\mylen

\usepackage{booktabs} 

\usepackage{caption} 
\begin{document}

\sloppy

\ifdefined\PAPER

\copyrightyear{2017}
\acmYear{2017}
\setcopyright{licensedothergov}
\acmConference{HPDC '17}{June 26-30, 2017}{Washington , DC, USA}\acmPrice{15.00}\acmDOI{http://dx.doi.org/10.1145/3078597.3078598}
\acmISBN{978-1-4503-4699-3/17/06}

\title{AllConcur: Leaderless Concurrent Atomic Broadcast}

\author{Marius Poke}
\affiliation{%
  \institution{HLRS \\ University of Stuttgart}
}
\email{marius.poke@hlrs.de}
\author{Torsten Hoefler}
\affiliation{%
  \institution{Department of Computer Science \\ ETH Zurich}
}
\email{htor@inf.ethz.ch}
\author{Colin W. Glass}
\affiliation{%
  \institution{HLRS\\University of Stuttgart}
}
\email{glass@hlrs.de}


\renewcommand{\shortauthors}{M. Poke, T. Hoefler, C. W. Glass}

\fi

\ifdefined\TECHREP
\title{\Large AllConcur: Leaderless Concurrent Atomic Broadcast \\ \large (Extended Version)\thanks{\textcopyright~2017 Copyright held by the owner/author(s). 
This is the author's version of the work. It is posted here for your personal use. Not for redistribution. 
The definitive version was published in HPDC~'17~\cite{poke2017allconcur}, http://dx.doi.org/10.1145/3078597.3078598.
Please refer to that publication when citing \LSC{}.}}

\author{
{\rm Marius Poke}\\
HLRS \\ 
University of Stuttgart \\
\emph{marius.poke@hlrs.de}
\and
{\rm Torsten Hoefler}\\
Department of Computer Science \\
ETH Zurich \\
\emph{htor@inf.ethz.ch}
\and
{\rm Colin W. Glass}\\
HLRS \\ 
University of Stuttgart\\
\emph{glass@hlrs.de}
} 

\maketitle
\thispagestyle{empty}

\fi

\begin{abstract}
Many distributed systems require coordination between the components involved. 
With the steady growth of such systems, the probability of failures increases, 
which necessitates scalable fault-tolerant agreement protocols. 
The most common practical agreement protocol, for such scenarios, is leader-based atomic broadcast. 
In this work, we propose \LSC{}, a distributed system that provides agreement through a leaderless 
concurrent atomic broadcast algorithm, thus, not suffering from the bottleneck of a central coordinator.
In \LSC{}, all components exchange messages concurrently through a logical overlay network that employs 
early termination to minimize the agreement latency. 
Our implementation of \LSC{} supports standard sockets-based TCP as well as high-performance InfiniBand Verbs communications. 
\LSC{} can handle up to $135$ million requests per second and achieves $17\times$ higher throughput than today's standard leader-based 
protocols, such as Libpaxos. 
%
Thus, \LSC{} is highly competitive with regard to existing solutions and, due to its decentralized approach, 
enables hitherto unattainable system designs in a variety of fields.
\end{abstract}

\ifdefined\PAPER
%
%
%


\maketitle
\fi

\section{Introduction}

Agreement is essential for many forms of collaboration in distributed systems. 
Although the nature of these systems may vary, ranging from distributed services provided by 
datacenters~\cite{Unterbrunner2014,Corbett:2012:SGG:2387880.2387905,DeCandia:2007:DAH:1323293.1294281} 
to distributed operating systems, such as Barrelfish~\cite{barrelfish} and Mesosphere's DC/OS~\cite{dcos}, they have in common that 
all the components involved regularly update a shared state.
In many applications, the state updates cannot be reduced, e.g., the actions of players in multiplayer video games.
Furthermore, the size of typical distributed systems has increased in recent years, 
making them more susceptible to single component failures~\cite{Sato:2012:DMN:2388996.2389022}.

Atomic broadcast is a communication primitive that provides fault-tolerant agreement 
while ensuring strong consistency of the overall system. 
\ifdefined\TECHREP
In a nutshell, it ensures that messages are received in the same order by all participants.
Atomic broadcast
\fi
\ifdefined\PAPER
It  
\fi
is often used to implement large-scale coordination services, 
such as replicated state machines~\cite{Hunt:2010:ZWC:1855840.1855851} or travel reservation systems~\cite{Unterbrunner2014}.
Yet, today’s practical atomic broadcast algorithms rely on leader-based approaches,  
\ifdefined\TECHREP
such as Paxos~\cite{Lamport:1998:PP:279227.279229,Lamport2001}.
In such algorithms, the ordering is ensured by a central coordinator, which may become a bottleneck, especially at large scale.
\fi
\ifdefined\PAPER
such as Paxos~\cite{Lamport:1998:PP:279227.279229,Lamport2001},
and thus, they may suffer from the bottleneck of a central coordinator, especially at large scale.
\fi

In this paper, we present \LSC{}\footnote{\LSCextended{}}---a distributed agreement system that relies 
on a leaderless atomic broadcast algorithm.
In \LSC{}, all participants exchange messages concurrently through an overlay network, described by a digraph 
\ifdefined\TECHREP
$G$~(\cref{sec:digraph}). 
\fi
\ifdefined\PAPER
$G$~(\cref{sec:Rbcast}). 
\fi 
The maximum number of failures \LSC{} can sustain is given by $G$'s connectivity and can be adapted to system-specific requirements~(\cref{sec:choose_g}).
Moreover, \LSC{} employs a novel early termination mechanism~(\cref{sec:earlyterm}) that reduces 
the expected number of communication steps significantly~(\cref{sec:prob_analysis}). 

\textbf{Distributed agreement vs. replication.} 
Distributed agreement is conceptually different from state machine replication
(SMR)~\cite{Schneider:1990:IFS:98163.98167,Lamport1978}:
Agreement targets collaboration in distributed systems, while SMR aims to increase data reliability.
Moreover, the number of agreeing components is an input parameter, while the number of replicas depends 
on the required data reliability.

\begin{figure}[!tp]
\centering
\subfloat[\label{fig:paxos}Leader-based agreement] {
\scalebox{0.9}{\input{figures/paxos.tex}}
}
\ifdefined\PAPER
\qquad
\fi
\subfloat[\label{fig:allconcur} \LSC{}]{        
\includegraphics[width=.17\textwidth]{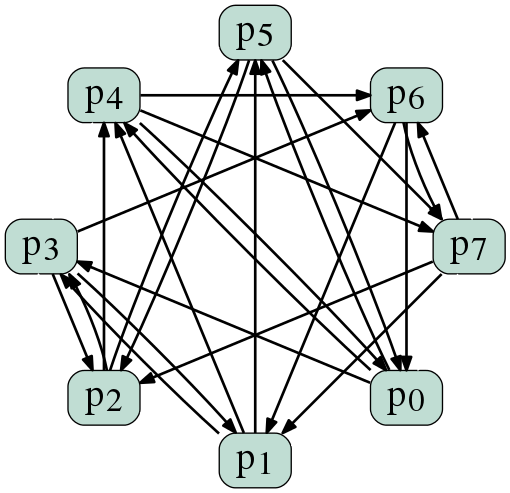}
}
\caption{Agreement among $8$ servers:
(a) Using a leader-based group; three operations needed per update---(1) send; (2) replicate; and (3) disseminate.
(b) Using a digraph $G$ with degree three and diameter two~\cite{Soneoka:1996:DDC:227095.227101}.}
\label{fig:deployment}
\end{figure}

\textbf{\LSC{} vs. leader-based agreement.} 
We consider the agreement among $n$ servers (see Figure~\ref{fig:deployment} for $n=8$). \LSC{} has the following properties:
(1) subquadratic work, i.e., $\mathcal{O}(nd)$, where $d$ is $G$'s degree~(\cref{sec:work}); 
2) adjustable depth, given by $G$'s diameter and fault diameter~(\cref{sec:prob_analysis}); 
(3) at most $2d$ connections per server; 
and (4) server-transitivity, i.e., all servers are treated equally, which entails fairness. 
In contrast, typical leader-based deployments do not have all of the above properties. 
Figure~\ref{fig:paxos} shows an example of leader-based agreement.
Each server sends updates to the group's leader, which, for reliability, replicates them within the group; 
the replicated updates are then disseminated to all servers. In typical leader-based approaches, 
such as client-server gaming platforms, servers interact directly with the leader (for both sending and receiving updates).
Although such methods have minimal depth and can ensure fairness, they require quadratic work and the leader needs 
to maintain $n$ connections~(\cref{sec:comparison}).

\textbf{Data consistency.}
\LSC{} provides agreement while guaranteeing strong consistency. In particular, we focus on the strong consistency of state updates; 
thus, throughout the paper we use both request and update interchangeably. 
For strongly consistent reads, queries also need to be serialized via atomic broadcast.
Serializing queries is costly, especially for read-heavy workloads. 
Typical coordination services~\cite{Hunt:2010:ZWC:1855840.1855851} 
relax the consistency model: Queries are performed locally and, hence, can return stale data.  
\LSC{} ensures that a server's view of the shared state cannot fall behind more than one round, i.e., one instance of concurrent atomic broadcast;  
thus, locally performed queries cannot be outdated by more than one round. 

\subsection{Applications and summary of results}
\LSC{} enables decentralized coordination services that require strong consistency 
at high request rates; thus, it allows for a novel approach to several real-world applications. 
We evaluate \LSC{} using a set of benchmarks, representative of 
\ifdefined\PAPER
two such applications: (1)~travel reservation systems; and (2)~multiplayer video games.
\fi
\ifdefined\TECHREP
three such applications: (1) travel reservation systems; (2) multiplayer video games; and (3) distributed exchanges.
\fi

\textbf{Travel reservation systems} 
are typical scenarios where updates are preceded by a large number of queries,
e.g., clients check many flights before choosing a ticket. 
To avoid overloading a central server, existing systems either adopt weaker consistency models, 
such as eventual consistency~\cite{DeCandia:2007:DAH:1323293.1294281}, or partition the state~\cite{Unterbrunner2014}, 
not allowing transactions spanning multiple partitions. 
\LSC{} offers strong consistency by distributing queries over multiple servers that agree on the entire state. 
Each server's rate of introducing updates in the system is bounded by its rate of answering queries. 
Assuming 64-byte updates, \LSC{} enables the agreement among $8$ servers, 
each generating 100 million updates per second, in $35\mu s$; 
moreover, the agreement among $64$ servers, each generating 32,000 updates per second, 
takes less than $0.75ms$. 

\textbf{Multiplayer video games}
are an example of applications where the shared state satisfies two 
conditions---it is too large to be frequently transferred through the network and it is periodically updated.
For example, modern video games update the state once every $50ms$ (i.e., $20$ frames per second) by 
only sending changes since the previous state~\cite{Bharambe:2008:DEL:1402958.1403002,Bharambe:2006:CDA:1267680.1267692}. 
Thus, such applications are latency sensitive~\cite{Beigbeder:2004:ELL:1016540.1016556}.
To decrease latency, existing systems either limit the number of players, e.g., $\approx8$ players in real time strategy games, 
or limit the players' view to only a subset of the state, such as the area of interest in first person shooter games~\cite{Bharambe:2008:DEL:1402958.1403002,Bharambe:2006:CDA:1267680.1267692}.
\LSC{} allows hundreds of servers to share a global state view at low latency; 
e.g., it supports the simultaneous interaction of $512$ players, 
using typical update sizes of 40 bytes~\cite{Bharambe:2008:DEL:1402958.1403002}, 
with an agreement latency of $38ms$, thus, enabling so called epic battles~\cite{epic_battles},
while providing strong consistency.

\ifdefined\TECHREP
\textbf{Distributed exchanges}
are a typical example of systems where fairness is essential.
For example, to connect to an exchange service, such as the New York Stock Exchange, 
clients must obtain so called co-locations (i.e., servers with minimal latency to the exchange). 
To ensure fairness, such co-locations are usually standardized; every client subscribing 
to the same co-location service has the same latency, ensured by standardized hardware~\cite{NasdaqRule}. 
Thus, centralized systems cannot support geographically distributed clients. 
\LSC{} enables the deployment of exchange services over geographically distributed servers: 
As long as all clients have an equal latency to any of the servers, fairness is provided. 
Assuming 40-byte client requests, an \LSC{} deployment across 8 servers can process 100 million 
requests per second with a median latency of less than $90\mu s$.  
\fi

In addition, \LSC{} can handle up to 135 million (8-byte) requests per second and achieves $17\times$ 
higher throughput than Libpaxos~\cite{libpaxos}, an implementation of Paxos~\cite{Lamport:1998:PP:279227.279229,Lamport2001},
while its average overhead of providing fault-tolerance is~$58\%$~(\cref{sec:lsc_eval}).

In summary, our work makes four key contributions:
\begin{itemize}
\setlength{\itemsep}{0pt}
  \item the design of \LSC{}---a distributed system that provides agreement 
  through a leaderless concurrent atomic broadcast algorithm~(\cref{sec:lsc_algo});
  \item a proof of \LSC{}'s correctness~(\cref{sec:lsc_correct});
  \item an analysis of \LSC{}'s performance~(\cref{sec:lsc_perform});
  \item implementations over standard sockets-based TCP and high-performance 
  InfiniBand Verbs, that allows us to evaluate \LSC{}'s performance~(\cref{sec:lsc_eval}).
\end{itemize}

\section{The broadcast problem}
\label{sec:bcast}

We consider $n$ servers connected through an overlay network, described by a digraph $G$. 
The servers communicate through messages, which cannot be lost (only delayed)---reliable communication.
Each server may fail according to a \emph{fail--stop} model: A server either operates 
correctly or it fails without further influencing other servers in the group. 
A server that did not fail is called \emph{non-faulty}.
We consider algorithms that tolerate up to $f$ failures, i.e., $f$-resilient.

In this paper, we use the notations from Chandra and Toueg~\cite{Chandra:1996:UFD:226643.226647} 
to describe both reliable and atomic broadcast: 
$m$ is a message (that is uniquely identified);
\rbroadcast{m}, \rdeliver{m}, \tobroadcast{m}, \todeliver{m} are communication primitives 
for broadcasting and delivering messages reliably \mbox{(\emph{R}-)} or atomically \mbox{(\emph{A}-)};
and \sender{m} is the server that R- or A-broadcasts~$m$.
Note that any message $m$ can be R- or A-broadcast at most once.

\subsection{Reliable broadcast}
\label{sec:Rbcast}

Any (non-uniform) reliable broadcast algorithm must satisfy 
three properties~\cite{Chandra:1996:UFD:226643.226647, Hadzilacos:1994:MAF:866693}:
\begin{itemize}
\setlength{\itemsep}{0pt}
  \item (Validity) If a non-faulty server R-broadcasts $m$, then it eventually R-delivers $m$.
  \item (Agreement) If a non-faulty server R-delivers $m$, then all non-faulty servers eventually 
  R-deliver $m$.
  \item (Integrity) For any message $m$, every non-faulty server R-delivers $m$ at most once, and only if 
  $m$ was previously R-broadcast by \sender{m}.
\end{itemize}

A simple reliable broadcast algorithm uses a complete digraph for message dissemination~\cite{Chandra:1996:UFD:226643.226647}. 
When a server executes \rbroadcast{m}, it sends $m$ to all other servers; when a server receives $m$ for the first 
time, it executes \rdeliver{m} only after sending $m$ to all other servers. 
Clearly, this algorithm solves the reliable broadcast problem. 
Yet, the all-to-all overlay network is unnecessary: For $f$-resilient reliable broadcast, 
it is sufficient to use a digraph with vertex-connectivity larger than $f$.

\begin{table}[!tp]
\centering
\ifdefined\PAPER
\small
\fi
\ifdefined\TECHREP
\scriptsize
\fi
\begin{tabular}{ l  l | l l }

  \cmidrule[1.5pt](){1-4}

  \textbf{Notation} &
  \textbf{Description} &
  \textbf{Notation} &
  \textbf{Description} \\
  \hline

  \rowcolor{DarkGray}  
  $G$ & 
  the digraph &  
  $d(G)$ & 
  degree \\
  
  $V(G)$ & 
  vertices  &
  $D(G)$ &
  diameter \\ 

  \rowcolor{DarkGray}  
  $E(G)$ & 
  directed edges &
  $\pi_{u,v}$ &
  path from $u$ to $v$ \\
    
  $v^+(G)$ &
  successors of $v$ &
  $k(G)$ &
  vertex-connectivity \\
  
  \rowcolor{DarkGray}  
  $v^-(G)$ &
  predecessors of $v$ &
  $D_f(G,f)$ &
  fault diameter \\

  \cmidrule[1.5pt](){1-4} 
\end{tabular}
  \caption{Digraph notations.}
\label{tab:notations}
\end{table}

\ifdefined\PAPER
\textbf{Fault-tolerant digraphs.}
We define a digraph $G$ by a set of $n$ vertices $V(G)=\left\{v_i:0\leq i \leq n-1\right\}$
and a set of directed edges $E(G) \subseteq \left\{(u,v):u,v\in V(G) \text{ and } u \neq v\right\}$.
In the context of reliable broadcast, the following parameters are of interest: 
the degree $d(G)$; the diameter $D(G)$; the vertex-connectivity $k(G)$;
and the fault diameter $D_f(G,f)$. The \emph{fault diameter} is the maximum diameter of $G$ after 
removing any $f<k(G)$ vertices~\cite{Krishnamoorthy:1987:FDI:35064.36256}. 
Also, we refer to digraphs with $k(G) = d(G)$ as \emph{optimally connected}~\cite{Meyer:1988:FFG:47054.47067, Dekker2004a}.
Finally, we use the following additional notations: 
$v^+(G)$ is the set of successors of $v \in V(G)$; 
$v^-(G)$ is the set of predecessors of $v \in V(G)$;
and $\pi_{u,v}$ is a path between two vertices $u,v \in V(G)$.
Table~\ref{tab:notations} summarizes all the digraph notations used throughout the paper.
\fi

\ifdefined\TECHREP
\subsubsection{Fault-tolerant digraphs}
\label{sec:digraph}

Let $G$ be a digraph with a set of $n$ vertices $V(G)=\left\{v_i:0\leq i \leq n-1\right\}$
and a set of directed edges $E(G) \subseteq \left\{(u,v):u,v\in V(G) \text{ and } u \neq v\right\}$.
Then $G$ has four parameters: (1) degree $d(G)$; (2) diameter $D(G)$; 
(3) vertex-connectivity $k(G)$; and (4) fault diameter $D_f(G,f)$.
Table~\ref{tab:notations} summarizes all the digraph notations used throughout the paper. 

\textbf{Degree.}
A digraph's degree relates to the concepts of both successor and predecessor of a 
vertex---given an edge $(u,v)\in E(G)$, $v$ is a \emph{successor} of $u$, 
while $u$ is a \emph{predecessor} of $v$. 
For a vertex $v$, the set of all successors (predecessors) is denoted by $v^+(G)$ ($v^-(G)$).
The out-degree (in-degree) of a vertex is the number of its successors (predecessors), 
i.e., $|v^+(G)|$ ($|v^-(G)|$). $G$'s \emph{degree}, denoted by $d(G)$, is the maximum in- or out-degree 
over all vertices; moreover, $G$ is \emph{$d$-regular} (or just regular) 
if $d(G)=|v^+(G)|=|v^-(G)|=d,\,\forall v\in V(G)$.

\textbf{Diameter.}
A path from vertex $u$ to vertex $v$ is a sequence of vertices $\pi_{u,v}=\left(v_{x_1},\ldots,v_{x_d}\right)$ 
that satisfies four conditions: (1) $v_{x_1}=u$; (2) $v_{x_d}=v$; 
(3) $v_{x_i}\neq v_{x_{j\neq i}},\,\forall i,j$; and (4) $(v_{x_i},v_{x_{i+1}})\in E,\,\forall 1 \leq i < d$.
The length of a path, denoted by $|\pi_{u,v}|$ is defined by the number of contained edges. 
$G$'s \emph{diameter}, denoted by $D(G)$, is the length of the longest shortest path 
between any two vertices.

\textbf{Connectivity.}
$G$ is connected if $\exists \pi_{u,v},\,\forall u\neq v\in V(G)$. 
\emph{Vertex-connectivity}, denoted by $k(G)$, is the minimum number of vertices whose removal 
results in a disconnected or a single-vertex digraph.
An alternative formulation is based on the notion of disjoint paths: 
Two paths are \emph{vertex-disjoint} if they contain no common internal vertices.
Thus, according to Menger's theorem, the vertex-connectivity equals the minimum number 
of vertex-disjoint paths between any two vertices.
Vertex-connectivity is bounded by the degree, i.e., $k(G) \leq d(G)$; 
digraphs with $k(G) = d(G)$ are said to be \emph{optimally connected}~\cite{Meyer:1988:FFG:47054.47067, Dekker2004a}.

\textbf{Fault diameter.}
Let $F \subset V(G)$ be a set of $f < k(G)$ vertices that are removed from $G$ resulting 
in a digraph $G_F$ with $V(G_F)=V(G)\setminus F$ and $E(G_F) = \left\{(u,v)\in E(G): u,v\in V(G_F)\right\}$.
For any subset $F$, the resulting digraph $G_F$ is connected; yet, the diameter of $G_F$ may be 
larger than $D(G)$. $G$'s \emph{fault diameter}, denoted by $D_f(G,f)$, is 
the maximum diameter of $G_F,\,\forall F \subset V(G)$.
\fi

\subsection{Atomic broadcast}
\label{sec:TObcast}

In addition to the reliable broadcast properties, atomic broadcast
must also satisfy the following property~\cite{Chandra:1996:UFD:226643.226647, Hadzilacos:1994:MAF:866693}:
\begin{itemize}
\setlength{\itemsep}{0pt}
  \item (Total order) If two non-faulty servers $p$ and $q$ A-deliver messages $m_1$ and $m_2$, 
  then $p$ A-delivers $m_1$ before $m_2$, if and only if $q$ A-delivers $m_1$ before $m_2$. 
\end{itemize}

There are different mechanisms to ensure total order~\cite{Defago:2004:TOB:1041680.1041682}. 
A common approach is to use a distinguished server (leader) as a coordinator. 
Yet, this approach suffers from the bottleneck of a central coordinator~(\cref{sec:comparison}).
An alternative entails broadcast algorithms that ensure atomicity through 
\emph{destinations agreement}~\cite{Defago:2004:TOB:1041680.1041682}: 
All non-faulty servers agree on a message set that is A-delivered.
%
Destinations agreement reformulates the atomic broadcast problem as \emph{consensus} problem~\cite[Chapter~5]{attiya2004distributed};
note that consensus and atomic broadcast are equivalent~\cite{Chandra:1996:UFD:226643.226647}.

\subsubsection{Lower bound}
\label{sec:lowerbound}
Consensus has a known synchronous lower bound: 
In a synchronous round-based model~\cite[Chapter~2]{attiya2004distributed}, any $f$-resilient consensus 
algorithm requires, in the worst case, at least $f+1$ rounds.
Intuitively, a server may fail after sending a message to only one other server; this scenario 
may repeat up to $f$ times, resulting in only one server having the message; this server needs 
at least one additional round to disseminate the message.
For more details, see the proof provided by Aguilera and Toueg~\cite{Aguilera:1998:SBP:866987}.
Clearly, if $G$ is used for dissemination, consensus requires (in the worst case) $f+D_f(G,f)$ rounds.
To avoid assuming always the worst case, we design an \emph{early termination} scheme~(\cref{sec:earlyterm}). 

\subsubsection{Failure detectors}

The synchronous model is unrealistic for real-world distributed systems; 
more fitting is to consider an asynchronous model.
Yet, under the assumption of failures, consensus (or atomic broadcast) cannot be 
solved in an asynchronous model~\cite{Fischer:1985:IDC:3149.214121}: 
We cannot distinguish between failed and slow servers.
To overcome this, we use a failure detector (FD). 
FDs are distributed oracles that provide information about faulty servers~\cite{Chandra:1996:UFD:226643.226647}. 

FDs have two main properties: \emph{completeness} and \emph{accuracy}. 
Completeness requires that all failures are eventually detected;
accuracy requires that no server is suspected to have failed before actually failing.
If both properties hold, then the FD is \emph{perfect} 
(denoted by $\mathcal{P}$)~\cite{Chandra:1996:UFD:226643.226647}.
In practice, completeness is easily guaranteed by a heartbeat mechanism: Each server 
periodically sends heartbeats to its successors; once it fails, its successors detect the lack of heartbeats. 

Guaranteeing accuracy in asynchronous systems is impossible---message delays are unbounded.
Yet, the message delays in practical distributed systems are bounded.
Thus, accuracy can be probabilistically guaranteed~(\cref{sec:probfd}).
Also, FDs can guarantee \emph{eventual accuracy}---eventually, 
no server is suspected to have failed before actually failing.
Such FDs are known as \emph{eventually perfect} (denoted by $\Diamond\mathcal{P}$)~\cite{Chandra:1996:UFD:226643.226647}.
For now, we consider an FD that can be reliably treated as $\mathcal{P}$. 
Later, we discuss the implications of using $\Diamond\mathcal{P}$, 
which can falsely suspect servers to have failed~(\cref{sec:eventualacc}).

\subsection{Early termination}
\label{sec:earlyterm}

The synchronous lower bound holds also in practice: 
A message may be retransmitted by $f$ faulty servers, before a non-faulty server can 
disseminate it completely. Thus, in the worst case, any $f$-resilient consensus algorithm that 
uses $G$ for dissemination requires $f+D_f(G,f)$ communication steps.
Yet, the only necessary and sufficient requirement for safe termination is for \emph{every non-faulty 
server to A-deliver messages only once it has all the messages any other non-faulty server has}.
Thus, early termination requires each server to track all the messages in the system.

\begin{figure*}[!tp]
\centering
\subfloat[\label{fig:bg_example}] {
\includegraphics[width=.2\textwidth]{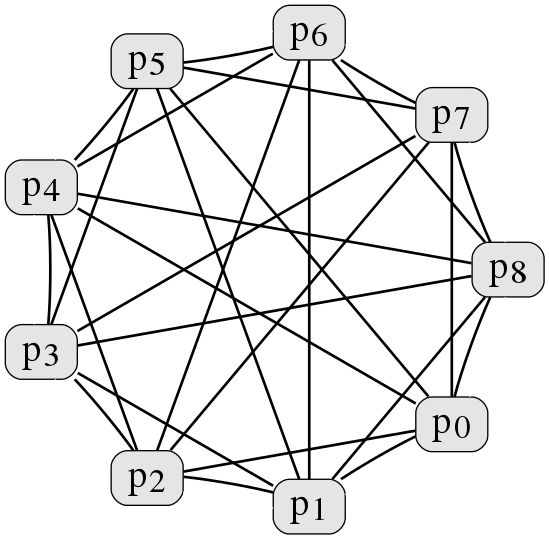}
}
\subfloat[\label{fig:tracking_digraph}] {
\ifdefined\PAPER
\scalebox{0.79}{\input{figures/tracking_digraph.tex}}
\fi
\ifdefined\TECHREP
\scalebox{0.72}{\input{figures/tracking_digraph.tex}}
\fi
}
\caption{(a) A binomial graph. (b) Message tracking within a binomial graph. Messages are shown chronologically from left to right. 
Dashed edges indicate failure notifications; for clarity, we omit the edges to the root of the digraphs.}
\label{fig:tracking_in_bg}
\end{figure*}

Early termination has two parts:
(1) deciding whether a message was A-broadcast; and (2) tracking the A-broadcast messages.
In general, deciding whether a message was A-broadcast entails waiting for $f+D_f(G,f)$ 
communication steps (the worst case scenario must be assumed for safety).
This essentially eliminates any form of early termination, if at least one server does not send a message. 
Yet, if every server A-broadcasts a message\footnote{The message can also be empty---the server A-broadcasts the information that it 
has nothing to broadcast.}, it is \emph{a priori} clear which messages exist; thus, no server waits for non-existent messages.

Every server tracks the A-broadcast messages through the received failure notifications.
As an example, we consider a group of nine servers that communicate via 
a binomial graph~\cite{Angskun2007}---a generalization of 1-way dissemination~\cite{Hensgen:1988:TAB:54616.54617}.
In binomial graphs, two servers $pi$ and $pj$ are connected if 
$j=i\pm2^l(\text{mod}\, n),\forall 0\leq l \leq \lfloor\log_2{n}\rfloor$ (see Figure~\ref{fig:bg_example}).
We also consider a failure scenario in which $p_0$ fails after sending 
its message $m_0$ only to $p_1$; 
$p_1$ receives $m_0$, yet, it fails before it can send it further.
How long should another server, e.g., $p_6$, wait for $m_0$? 

Server $p_6$ is not directly connected to $p_0$, so it cannot directly detect its failure. 
Yet, $p_0$'s non-faulty successors eventually detect $p_0$'s failure. Once they suspect 
$p_0$ to have failed, they stop accepting messages from $p_0$; 
also, they R-broadcast notifications of $p_0$'s failure. 
For example, let $p_6$ receive such a notification from $p_2$; 
then, $p_6$ knows that, if $p_2$ did not already send $m_0$, then $p_2$ did not receive $m_0$ from $p_0$.
Clearly, both \tobroadcast{} and \rbroadcast{} use the same paths for dissemination; 
the only difference between them is the condition to deliver a message. 
If $p_2$ had received $m_0$ from $p_0$, then it would 
have sent it to $p_6$ before sending the notification of $p_0$'s failure.
Thus, using failure notifications, $p_6$ can track the dissemination of $m_0$. 
Once $p_6$ receives failure notifications from all of $p_0$'s 
and $p_1$'s non-faulty successors, it knows that no non-faulty server is in possession of $m_0$. 

\section{The AllConcur algorithm}
\label{sec:lsc_algo}

\LSC{} is a completely decentralized, $f$-resilient, round-based atomic broadcast algorithm 
that uses a digraph $G$ as an overlay network. 
In a nutshell, in every round $R$, every non-faulty server performs three tasks: 
(1) it A-broadcasts a single (possibly empty) message;
(2) it tracks the messages A-broadcast in $R$ using the early termination mechanism described in Section~\ref{sec:earlyterm};
and (3) once done with tracking, it A-delivers---in a deterministic order---all the messages A-broadcast in $R$ that it received.
Note that A-delivering messages in a deterministic order entails that A-broadcast messages do not have to be received in the same order.
When a server fails, its successors detect the failure and R-broadcast failure notifications to the other servers; 
these failure notifications enable the early termination mechanism.
Algorithm~\ref{alg:lsc} shows the details of \LSC{} during a single round. 
Later, we discuss the requirements of iterating \LSC{}.

Initially, we make the following two assumptions: (1) the maximum number of failures is bounded, i.e., $f < k(G)$;  
and (2) the failures are detected by $\mathcal{P}$. 
In this context, we prove correctness---we show that the four properties of (non-uniform) atomic broadcast 
are guaranteed~(\cref{sec:lsc_correct}).
Then, we provide a probabilistic analysis of accuracy: If the network delays can be approximated as part of a known 
distribution, then we can estimate the probability of the accuracy property to hold~(\cref{sec:probfd}).
Finally, we discuss the consequences of dropping the two assumptions~(\cref{sec:consequences}). 

\LSC{} is message-based. Each server $p_i$ receives messages from its predecessors and sends messages to its successors.
We distinguish between two message types: 
(1) $\langle \mathit{BCAST},\, m_j \rangle$, a message A-broadcast by $p_j$;
and (2) $\langle \mathit{FAIL},\, p_j,\, p_k \in p_j^+(G)\rangle$, a notification, 
R-broadcast by $p_k$, indicating $p_k$'s suspicion that its predecessor $p_j$ has failed.
Note that if $p_i$ receives the notification and $p_k=p_i$, then it originated from $p_i$'s own FD.
Algorithm~\ref{alg:lsc} starts when at least one server A-broadcasts a message (line~\ref{alg:lsc_start}). 
Every server sends a message of its own, at the latest as a reaction upon receiving a message. 

\begin{algorithm}[!tb]
\scriptsize
\DontPrintSemicolon

\SetKwData{Input}{$M_i$}
\SetKwData{Fails}{$F_i$}
\SetKwData{aServer}{$p_\ast$}
\SetKwData{aMessage}{$m_\ast$}
\SetKwData{Suspected}{$\mathit{Q}$}
\SetKwData{InputGraph}{$\mathbf{g_i}$}
\SetKwData{FailDetector}{$\mathit{FD}$}

\SetKwData{Terminate}{check\_termination()}
\SetKwData{Done}{$\mathit{done}$}
\SetKwData{Continue}{$\mathit{continue}$}

\SetKwFunction{SendTo}{SendMsgTo}
\SetKwFunction{Sort}{sort}

\KwIn{$n$; $f$; $G$; $m_i$; $\Input\leftarrow\emptyset$; 
$\Fails\leftarrow\emptyset$;
$V(\InputGraph[p_i])\leftarrow\emptyset$; $V(\InputGraph[p_j])\leftarrow\{p_j\},\,\forall j \neq i$}

\BlankLine

\SetKwBlock{tobcast}{def \tobroadcast{m_i}:}{end}
\tobcast
{
 \label{alg:lsc_start}
  \textbf{send} $\langle \mathit{BCAST},\, m_i \rangle$ \textbf{to} $p_i^+(G)$\;
  \Input $\leftarrow$ \Input $\cup$ $\{m_i\}$\;
  \Terminate\;
}

\SetKwBlock{terminate}{def \Terminate:}{end}
\terminate
{
  \If{$V(\InputGraph[p])=\emptyset,\,\forall p$} 
  {\label{alg:lsc_termin_cond}
    \ForEach{$m \in \Sort{\Input}$} {
      \todeliver{m}
      \tcp*[r]{A-deliver messages}
      \label{alg:lsc_to_deliver}
     } 
     \tcc{preparing for next round}
     \ForEach{server \aServer} {\label{alg:lsc_server_remove} 
      \If{$\aMessage \notin \Input$}{
	$V(G) \leftarrow V(G) \setminus \{\aServer\}$ \tcp*[r]{remove servers}
      }
     }
     \ForEach{$(p,p_s) \in \Fails$ s.t. $p \in V(G) $} {\label{alg:lsc_fail_resend} 
      \textbf{send} $\langle \mathit{FAIL},\, p,\, p_s \rangle$ \textbf{to} $p_i^+(G)$ \tcp*[r]{resend failures}
     }
  }
}

\SetKwBlock{recvbcast}{receive $\langle \mathit{BCAST},\, m_j \rangle$:}{end}
\recvbcast
{
  \label{alg:lsc_recv_input}
  \lIf{$m_i \notin \Input$}{\tobroadcast{m_i}}
  \Input $\leftarrow$ \Input $\cup$ $\{m_j\}$\;
  \For{$m \in  \Input \text{ not already sent}$} {
    \textbf{send} $\langle \mathit{BCAST},\, m \rangle$ \textbf{to} $p_i^+(G)$
    \tcp*[r]{disseminate messages}
  }
  $V(\InputGraph[p_j]) \leftarrow \emptyset$\;
  \Terminate\;
}

\SetKwBlock{recvfail}{receive $\langle \mathit{FAIL},\, p_j,\, p_k \in p_j^+(G) \rangle$:}{end}
\recvfail
{
  \label{alg:lsc_recv_fail}
  \tcc{if $k=i$ then notification from local \FailDetector}
  \textbf{send} $\langle \mathit{FAIL},\, p_j,\, p_k \rangle$ \textbf{to} $p_i^+(G)$
  \tcp*[r]{disseminate failures}
  $\Fails \leftarrow \Fails \cup \{(p_j,p_k)\}$\;
  \ForEach{server \aServer}
  {
    \lIf{$p_j \notin V(\InputGraph[\aServer])$}{\Continue}
    
    \If{$p_j^+(\InputGraph[\aServer]) = \emptyset$}
    {\label{alg:lsc_no_succs}
      \tcc{maybe $p_j$ sent \aMessage to someone in $p_j^+(G)$ before failing}
      $\Suspected \leftarrow \{(p_j,p):p \in p_j^+(G) \setminus \{p_k\}\}$
      \tcp*[r]{FIFO queue}
      \ForEach{$(p_p,p) \in \Suspected$} {
          $\Suspected \leftarrow \Suspected \setminus \{(p_p,p)\}$\;
          \If{$p \notin V(\InputGraph[\aServer])$}
          {
	    $V(\InputGraph[\aServer]) \leftarrow V(\InputGraph[\aServer]) \cup \{p\}$\;\label{alg:lsc_add_vertex}
	    \If{$\exists (p,*) \in \Fails$} 
	    {\label{alg:lsc_already_failed}
	      $\Suspected \leftarrow \Suspected \cup \{(p,p_s):p_s \in p^+(G)\}\setminus\Fails$
	    }
          }
          $E(\InputGraph[\aServer]) \leftarrow E(\InputGraph[\aServer]) \cup \{(p_p,p)\}$\;\label{alg:lsc_add_edge}
      }       
    }
    \ElseIf{$p_k \in p_j^+(\InputGraph[\aServer])$}
    {\label{alg:lsc_succ_pk}
      \tcc{$p_k$ has not received \aMessage from $p_j$}
      $E(\InputGraph[\aServer]) \leftarrow E(\InputGraph[\aServer]) \setminus \{(p_j,p_k)\}$\;\label{alg:lsc_rm_edge}
      \ForEach{$p \in V(\InputGraph[\aServer])$ s.t. $\nexists \pi_{p_\ast,p}$ in $\InputGraph[\aServer]$} 
      {\label{alg:lsc_no_input}
	$V(\InputGraph[\aServer]) \leftarrow V(\InputGraph[\aServer]) \setminus \{p\}$
	\tcp*[r]{no input}
      }
    }  
    
    \If{$\forall p\in V(\InputGraph[\aServer]),\, (p,*)\in\Fails$}
    {\label{alg:lsc_no_dissemination}
      $V(\InputGraph[\aServer]) \leftarrow \emptyset$
      \tcp*[r]{no dissemination}
    }
  }  
  \Terminate\;
}
\caption{The \LSC{} algorithm; code executed by server $p_i$; see Table~\ref{tab:notations} for digraph notations.}
\label{alg:lsc}
\end{algorithm}

\textbf{Termination.}
\LSC{} adopts a novel early termination mechanism~(\cref{sec:earlyterm}).
To track the A-broadcast messages, each server $p_i$ stores an array $\mathbf{g_i}$
of $n$ digraphs, one for each server $p_\ast\in V(G)$; we refer to these as \emph{tracking digraphs}.
The vertices of each tracking digraph $\mathbf{g_i}[p_\ast]$ consist of the servers which (according to $p_i$) 
may have $m_\ast$. 
An edge $(p_j,p_k) \in E(\mathbf{g_i}[p_\ast])$ indicates $p_i$'s suspicion that $p_k$ received $m_\ast$ directly from $p_j$.
If $p_i$ has $m_\ast$, then $\mathbf{g_i}[p_\ast]$ is no longer needed; hence, $p_i$ removes all its vertices, 
i.e., $V(\mathbf{g_i}[p_\ast])=\emptyset$.
Initially, $V(\mathbf{g_i}[p_j])=\{p_j\},\,\forall p_j\neq p_i$ and $V(\mathbf{g_i}[p_i])=\emptyset$.
Server $p_i$ A-delivers all known messages (in a deterministic order) once all tracking digraphs are empty (line~\ref{alg:lsc_to_deliver}).

Figure \ref{fig:tracking_digraph} illustrates the message-driven changes to the tracking digraphs based on the binomial graph example in Section~\ref{sec:earlyterm}.
For clarity, we show only two of the messages being tracked by server $p_6$ (i.e., $m_0$ and $m_1$); both messages are tracked 
by updating $\mathbf{g_6}[p_0]$ and $\mathbf{g_6}[p_1]$, respectively. 
First, $p_6$ receives from $p_2$ a notification of $p_0$'s failure, which indicates that $p_2$ has not received $m_0$ directly from $p_0$~(\cref{sec:earlyterm}).
Yet, $p_0$ may have sent $m_0$ to its other successors; thus, $p_6$ adds them to $\mathbf{g_6}[p_0]$.
Next, $p_6$ receives from $p_5$ a notification of $p_0$'s failure---$p_5$ has not received $m_0$ directly from $p_0$ either and, thus, the edge $(p_0,p_5)$ is removed. 
Then, $p_6$ receives from $p_3$ a notification of $p_1$'s failure. Hence, $p_6$ extends both $\mathbf{g_6}[p_0]$ and $\mathbf{g_6}[p_1]$ with $p_1$'s successors (except $p_3$).
In addition, due to the previous notifications of $p_0$'s failure, $p_6$ extends $\mathbf{g_6}[p_1]$ with $p_0$'s successors (except $p_2$ and $p_5$).
Finally, $p_6$ receives $m_1$; thus, it removes all the vertices from $\mathbf{g_6}[p_1]$ (i.e., it stops tracking $m_1$).

\textbf{Receiving $\langle \mathit{BCAST},\, m_j \rangle$.}
When receiving an A-broadcast message $m_j$ (line~\ref{alg:lsc_recv_input}), server $p_i$ adds it to the 
set $M_i$ of known messages. Also, it A-broadcasts its own message $m_i$, in case it did not do so before. 
Then, it continues the dissemination of each known message through the network---$p_i$ sends 
all unique messages it has not already sent to its successors $p_i^+(G)$. Finally, $p_i$ removes all the vertices 
from the $\mathbf{g_i}[p_j]$ digraph; then, it checks whether the termination conditions are fulfilled.

\textbf{Receiving $\langle \mathit{FAIL},\, p_j,\, p_k\rangle$.}
When receiving a notification, R-broad\-cast by $p_k$, indicating $p_k$'s suspicion that $p_j$ has failed 
(line~\ref{alg:lsc_recv_fail}), $p_i$ disseminates it further.
Then, it adds a tuple $(p_j,p_k)$ to the set $F_i$ of received failure notifications. 
Finally, it updates the tracking digraphs in $\mathbf{g_i}$ that contain $p_j$ as a vertex. 

We distinguish between two cases, depending on whether this is the first notification of $p_j$'s failure received by $p_i$. 
If it is the first,  $p_i$  updates all $\mathbf{g_i}[p_\ast]$ containing $p_j$ as a vertex by adding $p_j$'s successors (from $G$) together 
with the corresponding edges. The rationale is, that $p_j$ may have sent $m_\ast$ to his successors, who are now in possession of it. 
Thus, we track the possible whereabouts of messages.
However, there are some exceptions: 
Server $p_k$ could not have received $m_\ast$ directly from $p_j$~(\cref{sec:earlyterm}).
Also, if a successor $p_f \notin V(\mathbf{g_i}[p_\ast])$ is added, which is already known to have failed, 
it  may have already received $m_\ast$ and sent it further. 
Hence, the successors of $p_f$ could be in possession of $m_\ast$ and are added to $\mathbf{g_i}[p_\ast]$ in the same way as described above (line~\ref{alg:lsc_already_failed}).

If $p_i$ is already aware of $p_j$'s failure (i.e., the above process already took place),
the new failure notification informs $p_i$, that $p_k$ (the origin of the notification) has not received $m_\ast$ 
from $p_j$---because $p_k$ would have sent it before sending the failure notification. 
Thus, the edge ($p_j$, $p_k$) can be removed from $\mathbf{g_i}[p_\ast]$ (line~\ref{alg:lsc_succ_pk}).

In the end, $p_i$ prunes $\mathbf{g_i}[p_\ast]$ by removing the servers no longer of interest in tracking $m_\ast$. 
First, $p_i$ removes every server $p$ for which there is no path (in $\mathbf{g_i}[p_\ast]$) from $p_\ast$ to $p$, 
as $p$ could not have received $m_\ast$ from any of the servers in $\mathbf{g_i}[p_\ast]$ (line~\ref{alg:lsc_no_input}).
Then, if $\mathbf{g_i}[p_\ast]$ contains only servers already known to have failed, $p_i$ prunes it entirely---no 
non-faulty server has $m_\ast$ (line~\ref{alg:lsc_no_dissemination}).

\textbf{Iterating \LSC{}.}
Executing subsequent rounds of \LSC{} requires the correct handling of failures. 
Since different servers may end and begin rounds at different times, \LSC{} employs a consistent mechanism of tagging 
servers as failed: At the end of each round, all servers whose messages were not A-delivered are tagged as failed by 
all the other servers (line~\ref{alg:lsc_server_remove}). 
As every non-faulty server agrees on the A-delivered messages, this ensures a consistent view of failed servers. 
In the next round, every server resends the failure notifications, except those of 
servers already tagged as failed (line~\ref{alg:lsc_fail_resend}). 
Thus, only the tags and the necessary resends need to be carried over from the previous round.
Moreover, each message contains the sequence number~$R$ of the round in which it was first sent.
Thus, all messages can be uniquely identified, i.e., $\langle \mathit{BCAST},\, m_j \rangle$ by $(R, p_j)$ tuples 
and $\langle \mathit{FAIL},\, p_j,\, p_k\rangle$ by $(R, p_j, p_k)$ tuples,
which allows for multiple rounds to coexist.  

\textbf{Initial bootstrap and dynamic membership.}
To bootstrap \LSC{}, we require a centralized service, such as  ZooKeeper~\cite{Hunt:2010:ZWC:1855840.1855851}:
The system must decide on the initial configuration---the identity of the $n$ servers, the fault tolerance $f$ 
and the digraph $G$.
Once \LSC{} starts, any further reconfigurations are agreed upon via atomic broadcast.
This includes topology reconfigurations and membership changes, i.e., servers leaving and joining the system. 
In contrast to leader-based approaches, where such changes may necessitate a leader election, in \LSC{}, 
dynamic membership is handled directly by the algorithm.

\subsection{Correctness}
\label{sec:lsc_correct}

To prove \LSC{}'s correctness, we show that the four properties of 
(non-uniform) atomic broadcast are guaranteed~(\cref{sec:TObcast}).
Clearly, the integrity property holds: Every server $p_i$ executes \todeliver{} only once for 
each message in the set $M_i$, which contains only messages A-broadcast by some servers.
To show that the validity property holds, it is sufficient to prove that the algorithm 
terminates
(see Lemma~\ref{lemma:lsc_termination}). 
To show that both the agreement and the total order properties hold, 
it is sufficient to prove \emph{set agreement}---when the algorithm terminates, 
all non-faulty servers have the same set of known messages (see Lemma~\ref{lemma:lsc_agreement}).
To prove termination and set agreement, we introduce the following lemmas:

\begin{lemma}
\label{lemma:lsc_transmit_path}
Let $p_i$ be a non-faulty server; let $p_j \neq p_i$ be a server; 
let $\pi_{p_j,p_i}=\left(a_{1},\ldots,a_{d}\right)$ be a path (in digraph $G$) 
from $p_j$ to $p_i$. If $p_j$ knows a message $m$ (either its own or received), 
then, $p_i$ eventually receives either $\langle \mathit{BCAST},\, m \rangle$ 
or $\langle \mathit{FAIL},\, a_k,\, a_{k+1}\rangle$ with $1 \leq k < d$.
\end{lemma}
\begin{proof}
Server $p_j$ can either fail or send $m$ to $a_2$. 
Further, for each inner server $a_{k} \in \pi_{p_j,p_i}, 1<k<d$, we distinguish three scenarios:
(1) $a_{k}$ fails; (2) $a_{k}$ detects the failure of its predecessor on the path; 
or (3) $a_{k}$ further sends the message received from its predecessor on the path.
The message can be either $\langle \mathit{BCAST},\, m \rangle$ or 
$\langle \mathit{FAIL},\, a_l,\, a_{l+1}\rangle$ with $1 \leq l < k$.
Thus, $p_i$ eventually receives either $\langle \mathit{BCAST},\, m \rangle$ or 
$\langle \mathit{FAIL},\, a_k,\, a_{k+1}\rangle$ with $1 \leq k < d$.
Figure~\ref{fig:input_spread} shows, in a tree-like fashion, what messages can be transmitted 
along a three-server path.
\end{proof}

\begin{figure}[!tp]
\centering
\scalebox{0.65}{\input{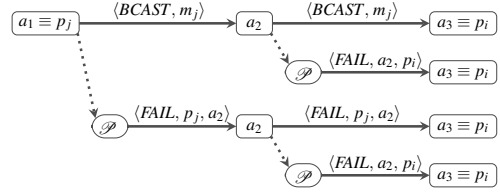}}
\caption{Possible messages along a three-server path. Dotted arrows indicate failure detection.}
\label{fig:input_spread}
\end{figure}

\begin{lemma}
\label{lemma:lsc_transmit}
Let $p_i$ be a non-faulty server; let $p_j \neq p_i$ be a server. 
If $p_j$ knows a message $m$ (either its own or received), 
then $p_i$ eventually receives either the message $m$ or a notification of $p_j$'s failure.   
\end{lemma}
\begin{proof}
If $p_i$ receives $m$, then the proof is done. In the case $p_i$ does not receive $m$, we assume it 
does not receive a notification of $p_j$'s failure either. 
Due to $G$'s vertex-connectivity, there are at least $k(G)$ vertex-disjoint paths $\pi_{p_j,p_i}$.
For each of these paths, $p_i$ must receive notifications of some inner vertex failures (cf.~Lemma~\ref{lemma:lsc_transmit_path}). 
Since the paths are vertex-disjoint, each notification indicates a different server failure. 
However, this contradicts the assumption that $f < k(G)$.
\end{proof}

\begin{corollary}
\label{corollary:lsc_transmit}
Let $p_i$ be a non-faulty server; let $p_j \neq p_i$ be a server.
If $p_j$ receives a message, then $p_i$ eventually receives either the same message or 
a notification of $p_j$'s failure.
\end{corollary}

\begin{lemma}
\label{lemma:lsc_rm_edges}
Let $p_i$ be a server; let $\mathbf{g_i}[p_j]$ be a tracking digraph that can no longer be pruned.
If $E(\mathbf{g_i}[p_j]) \neq \emptyset$, then $p_i$ eventually removes an edge from $E(\mathbf{g_i}[p_j])$.
\end{lemma}
\begin{proof}
We assume that $p_i$ removes no edge from $E(\mathbf{g_i}[p_j])$.
Clearly, the following statements are true: 
(1) $V(\mathbf{g_i}[p_j]) \neq \emptyset$ (since $E(\mathbf{g_i}[p_j]) \neq \emptyset$); 
(2) $p_j \in V(\mathbf{g_i}[p_j])$ (since $\mathbf{g_i}[p_j]$ can no longer be pruned); 
and (3) $p_j$ is known to have failed (since $V(\mathbf{g_i}[p_j])\neq\{p_j\}$).
Let $p \in V(\mathbf{g_i}[p_j])$ be a server such that $p_i$ receives no notification of $p$'s failure.
The reason $p$ exists is twofold: (1) the maximum number of failures is bounded; and 
(2) $\mathbf{g_i}[p_j]$ can no longer be pruned (line~\ref{alg:lsc_no_dissemination}).
Then, we can construct a path $\pi_{p_j,p}=\left(a_{1},\ldots,a_{d}\right)$ in $\mathbf{g_i}[p_j]$ 
such that every server along the path, except for $p$, is known to have failed (line~\ref{alg:lsc_no_input}).
Eventually, $p$ receives either $\langle \mathit{BCAST},\, m_j \rangle$ or 
$\langle \mathit{FAIL},\, a_k,\, a_{k+1}\rangle$ with $1 \leq k < d$ (cf.~Lemma~\ref{lemma:lsc_transmit_path}).
Since $p_i$ receives no notification of $p$'s failure, the message received by $p$ eventually arrives at $p_i$ 
(cf.~Corollary~\ref{corollary:lsc_transmit}).
On the one hand, if $p_i$ receives $\langle \mathit{BCAST},\, m_j \rangle$, then all edges are removed from $E(\mathbf{g_i}[p_j])$;
this leads to a contradiction.
On the other hand, if $p_i$ receives $\langle \mathit{FAIL},\, a_k,\, a_{k+1}\rangle$, then the edge $(a_k,a_{k+1})$
is removed from $E(\mathbf{g_i}[p_j])$ (line~\ref{alg:lsc_rm_edge}); this also leads to a contradiction.
\end{proof}

\begin{lemma}
\label{lemma:lsc_termination}
(Termination) Let $p_i$ be a non-faulty server. Then, $p_i$ eventually terminates.
\end{lemma}
\begin{proof}
If $V(\mathbf{g_i}[p])=\emptyset,\,\forall p$, then the proof is done (line~\ref{alg:lsc_termin_cond}).
We assume $\exists p_j$ such that $V(\mathbf{g_i}[p_j]) \neq \emptyset$ and $\mathbf{g_i}[p_j]$ can no longer be pruned. 
Clearly, $p_j \in V(\mathbf{g_i}[p_j])$. 
Server $p_i$ receives either $m_j$ or a notification of $p_j$'s failure (cf.~Lemma~\ref{lemma:lsc_transmit}).
If $p_i$ receives $m_j$, then all servers are removed from $V(\mathbf{g_i}[p_j])$, 
which contradicts $V(\mathbf{g_i}[p_j]) \neq \emptyset$.
We assume $p_i$ receives a notification of $p_j$'s failure; then, $p_j^+(\mathbf{g_i}[p_j]) \neq \emptyset$ 
(since $\mathbf{g_i}[p_j]$ can no longer be pruned); also, $E(\mathbf{g_i}[p_j]) \neq \emptyset$.
By repeatedly applying the result of Lemma~\ref{lemma:lsc_rm_edges}, it results that $p_i$ eventually 
removes all edges from $\mathbf{g_i}[p_j]$. As a result, $\mathbf{g_i}[p_j]$ is eventually completely pruned, 
which contradicts $V(\mathbf{g_i}[p_j]) \neq \emptyset$.
\end{proof}

\begin{lemma}
\label{lemma:lsc_agreement}
(Set agreement) Let $p_i$ and $p_j$ be any two non-faulty servers. 
Then, after \LSC{}'s termination, $M_i=M_j$.
\end{lemma}
\begin{proof}
It is sufficient to show that if $m_\ast \in M_i$ when $p_i$ terminates, then also $m_\ast \in M_j$ when $p_j$ terminates.
We assume that $p_j$ does not receive $m_\ast$.
Let $\pi_{p_\ast,p_i}=\left(a_{1},\ldots,a_{d}\right)$ be one of the paths (in $G$) on which $m_\ast$ arrives at $p_i$.
Let $k,\, 1\leq k \leq d$ the smallest index such that $p_j$ receives no notification of $a_k$'s failure.
The existence of $a_k$ is given by the existence of $p_i$, a server that is both non-faulty and on $\pi_{p_\ast,p_i}$.
Clearly, $a_k \in V(\mathbf{g_j}[p_\ast])$. Since it terminates, $p_j$ eventually removes 
$a_k$ from $\mathbf{g_j}[p_\ast]$. In general, $p_j$ can remove $a_k$ when it receives either $m_\ast$ or a notification of $a_k$'s 
failure; yet, both alternatives lead to contradictions. 
In addition, for $k > 1$, $p_j$ can remove $a_k$ when there is no path $\pi_{p_\ast,a_k}$ in $\mathbf{g_j}[p_\ast]$.
This requires $p_j$ to remove an edge on the $\left(a_{1},\ldots,a_{k}\right)$ path. 
Thus, $p_j$ receives a message $\langle \mathit{FAIL},\, a_l,\, a_{l+1}\rangle$ with $1 \leq l < k$.
Yet, since $a_{l+1}$ received $m_\ast$ from $a_{l}$, $p_j$ must receive $\langle \mathit{BCAST},\, m_\ast \rangle$ first, 
which leads to a contradiction.
\end{proof}

\begin{corollary}
\LSC{} solves the atomic broadcast problem while tolerating up to $f$ failures.
\end{corollary}

\subsection{Probabilistic analysis of accuracy}
\label{sec:probfd}

Algorithm~\ref{alg:lsc} assumes a perfect FD, which requires the accuracy property to hold. 
Accuracy is difficult to guarantee in practice: Due to network delays, 
a server may falsely suspect another server to have failed. 
Yet, when the network delays can be approximated as part of a known distribution, 
accuracy can be probabilistically guaranteed. 
Let $T$ be a random variable that describes the network delays. 
Then, we denote by $\mathit{Pr}[T > t]$ the probability that a message delay exceeds a constant $t$.

We propose an FD based on a heartbeat mechanism.
Every non-faulty server sends heartbeats to its successors in $G$;
the heartbeats are sent periodically, with a period $\Delta_\mathit{hb}$. 
Every non-faulty server $p_i$ waits for heartbeats from its predecessors in $G$;
if, within a period $\Delta_\mathit{to}$, $p_i$ receives no heartbeats from a predecessor $p_j$, 
it suspects $p_j$ to have failed. 
Since we assume heartbeat messages are delayed according to a known distribution, we can estimate the 
probability of the FD to be accurate, in particular a lower bound 
of the probability of the proposed FD to behave indistinguishably from a perfect one. 

The interval in which $p_i$ receives two heartbeats from a predecessor $p_j$ is bounded by 
$\Delta_\mathit{hb} + T$. In the interval $\Delta_\mathit{to}$, $p_j$ sends 
$\left\lfloor \Delta_\mathit{to} / \Delta_\mathit{hb} \right\rfloor$ heartbeats to $p_i$.
The probability that $p_i$ does not receive the $k$'th heartbeat within the period $\Delta_\mathit{to}$ 
is bounded by $\mathit{Pr}[T > \Delta_\mathit{to} - k\Delta_\mathit{hb}]$. 
For $p_i$ to incorrectly suspect $p_j$ to have failed, it has to receive none of the $k$ heartbeats.
Moreover, $p_i$ can incorrectly suspect $d(G)$ predecessors; also, there are $n$ servers 
that can incorrectly suspect their predecessors. 
Thus, the probability of the accuracy property to hold is at least
$
(1-
\prod\limits_{k=1}^{\lfloor\Delta_\mathit{to}/\Delta_\mathit{hb}\rfloor}
\mathit{Pr}[T > \Delta_\mathit{to} - k\Delta_\mathit{hb}])^{n \cdot d(G)}.
$

Increasing both the timeout period and the heartbeat frequency increases the likelihood 
of accurate failure detection. The probability of no incorrect failure detection in the system, 
together with the probability of less than $k(G)$ failures define the reliability of \LSC{}. 

\subsection{Widening the scope}
\label{sec:consequences}

A practical atomic broadcast algorithm must always guarantee safety.
Under the two initial assumptions, i.e., $f < k(G)$ and $\mathcal{P}$, \LSC{} guarantees both safety and liveness~(\cref{sec:lsc_correct}). 
In this section, we show that $f < k(G)$ is not required for safety, but only for liveness~(\cref{sec:discong}). 
Also, we provide a mechanism that enables \LSC{} to guarantee safety even when the $\mathcal{P}$ assumption is dropped~(\cref{sec:eventualacc}). 

\subsubsection{Disconnected digraph}
\label{sec:discong}

In general, Algorithm~\ref{alg:lsc} requires $G$ to be connected. 
A digraph can be disconnected by either (1) removing a sufficient number of vertices to break the vertex-connectivity, i.e., $f \geq k(G)$, 
or (2) removing sufficent edges to break the edge-connectivity. 
Under the assumption of reliable communication (i.e., $G$'s edges cannot be removed), only the fist scenario is possible. 
If $f \geq k(G)$, termination is not guaranteed (see Lemma~\ref{lemma:lsc_transmit}).
Yet, some servers may still terminate the round even if $G$ is disconnected.
In this case, set agreement still holds, as the proof of Lemma~\ref{lemma:lsc_agreement} 
does not assume less than $k(G)$ failures.
In summary, the $f < k(G)$ assumption is needed only to guarantee liveness; 
safety is guaranteed regardless of the number of failures (similar to Paxos~\cite{Lamport:1998:PP:279227.279229,Lamport2001}).

In scenarios where $G$'s edges can be removed, such as network partitioning, 
a non-faulty server disconnected from one of its non-faulty successors will be falsely suspected 
to have failed\footnote{Note that if $G$ is disconnected by removing vertices, 
a non-faulty server cannot be disconnected from its non-faulty successors.}.
Thus, the assumption of $\mathcal{P}$ does not hold and we need to relax it to~$\Diamond\mathcal{P}$.


\subsubsection{Eventual accuracy}
\label{sec:eventualacc}

For some distributed systems, it may be necessary to use $\Diamond\mathcal{P}$ instead of $\mathcal{P}$.
For instance, in cases of network partitioning as discussed above, or 
for systems in which approximating network delays as part of a known distribution is difficult. 
Implementing a heartbeat-based $\Diamond\mathcal{P}$ is straightforward~\cite{Chandra:1996:UFD:226643.226647}: 
When a server falsely suspects another server to have failed, it increments 
the timeout period $\Delta_\mathit{to}$; thus, eventually, non-faulty servers 
are no longer suspected to have failed.
%
Yet, when using $\Diamond\mathcal{P}$, failure notifications no longer necessarily indicate server failures. 
Thus, to adapt Algorithm~\ref{alg:lsc} to $\Diamond\mathcal{P}$, we need to ensure the correctness 
of early termination, which relies on the information carried by failure notifications. 

First, a $\langle \mathit{FAIL},\, p_j,\, p_k\rangle$ message received by $p_i$, 
indicates that $p_k$ did not receive (and it will not receive until termination) from $p_j$ 
any message not yet received by $p_i$. 
Thus, once a server suspects one of its predecessors to have failed, it must ignore any subsequent messages 
(except failure notifications) received from that predecessor (until the algorithm terminates).
As a result, when using $\Diamond\mathcal{P}$, it is still possible to decide if a server received a certain message. 

Second, $p_i$ receiving notifications of $p_j$'s failure from all $p_j$'s successors indicates 
both that $p_j$ is faulty and that it did not disseminate further any message not yet received by $p_i$.
Yet, when using $\Diamond\mathcal{P}$, these notifications no longer indicate that $p_j$ is faulty. 
Thus, both $p_i$ and $p_j$ can terminate without agreeing on the same set (i.e., $M_i \neq M_j$), which breaks \LSC{}'s safety.
In this case though, $p_i$ and $p_j$ are part of different strongly connected components.
For set agreement to hold~(\cref{sec:lsc_correct}), only the servers from one single strongly 
connected component can A-deliver messages; we refer to this component as the \emph{surviving partition}.
The other servers are considered to be faulty (for the properties of reliable broadcast to hold). 
To ensure the uniqueness of the surviving partition, it must contain at least a majority of the servers.

\textbf{Deciding whether to A-deliver.}
Each server decides whether it is part of the surviving partition via 
a mechanism based on Kosaraju's algorithm to find strongly connected components~\cite[Chapter 6]{Aho:1983:DSA:577958}. 
In particular, once each server $p_i$ decides on the set $M_i$, it R-broadcasts two messages:
(1) a forward message $\langle \mathit{FWD},\, p_i \rangle$; and (2) a backward message 
$\langle \mathit{BWD},\, p_i \rangle$. The backward message is R-broadcast using the transpose of $G$.
Then, $p_i$ A-delivers the messages from $M_i$ only when it receives both forward and backward messages 
from at least $\lfloor n/2 \rfloor$ servers.
Intuitively, a $\langle \mathit{FWD},\, p_j \rangle$ message received by $p_i$ indicates 
that when $p_j$ decided on its set $M_j$, there was at least one path from $p_j$ to $p_i$; 
thus, $p_i$ knows of all the messages known by $p_j$ (i.e., $M_j \subseteq M_i$). 
Similarly, a $\langle \mathit{BWD},\, p_j \rangle$ message indicates that $M_i \subseteq M_j$.
Thus, when $p_i$ A-delivers it knows that at least a majority of the servers (including itself) 
A-deliver the same messages.

\textbf{Non-terminating servers.}
To satisfy the properties of reliable broadcast~(\cref{sec:Rbcast}), 
non-terminating servers need to be eventually removed from the system and consequently, be considered as faulty. 
In practice, these servers could restart after a certain period of inactivity and 
then try to rejoin the system, by sending a membership request to one of the non-faulty servers.

\section{Performance analysis}
\label{sec:lsc_perform}

\LSC{} is designed as a high-throughput atomic broadcast algorithm. 
Its performance is given by 
three metrics: (1) work per server; (2) communication time; and (3) storage requirements.
Our analysis focuses on Algorithm~\ref{alg:lsc}, i.e., connected digraph and perfect FD, and 
it uses the LogP model~\cite{Culler:1993:LTR:173284.155333}. 
The LogP model is described by four parameters: the latency $L$; the overhead $o$; 
the gap between messages $g$; and the number of processes (or servers) $P$, which we denote by $n$.
We make the common assumption that $o > g$~\cite{Alexandrov:1995:LIL:215399.215427}; 
also, the model assumes short messages.
\LSC{}'s performance depends on $G$'s parameters: $d$, $D$, and $D_f$.
A discussion on how to choose $G$ is provided in Section~\ref{sec:choose_g}.

\subsection{Work per server}
\label{sec:work}

The amount of work a server performs is given by the number of messages it receives and sends. 
\LSC{} distinguishes between A-broadcast messages and failure notifications.
First, without failures, every server receives an A-broadcast message from all of its $d$ predecessors,
i.e., $(n-1) \cdot d$ messages received by each server.
This is consistent with the $\Omega(n^2f)$ worst-case message complexity for 
synchronous $f$-resilient consensus algorithms~\cite{Dolev:2013:ECE:2484239.2484269}.
Second, every failed server is detected by up to $d$ servers, 
each sending a failure notification to its $d$ successors. 
Thus, each server receives up to $d^2$ notifications of each failure. 
Overall, each server receives at most $n \cdot d + f \cdot d^2$ messages. 
Since $G$ is regular, each server sends the same number of messages.
 
In order to terminate, in a non-failure scenario, a server needs to receive at least $(n-1)$ messages 
and send them further to $d$ successors. We estimate the time of sending or receiving a message 
by the overhead $o$ of the LogP model~\cite{Culler:1993:LTR:173284.155333}.
Thus, a lower bound on termination (due to work) is given by $2(n-1)do$. 

\subsection{Communication time}
\label{sec:comm_time}

In general, the time to transmit a message (between two servers) is estimated by $T(\mathit{msg}) = L + 2o$.
We consider only the scenario of a single non-empty message $m$ being A-broadcast and 
we estimate the time between \sender{m} A-broadcasts $m$ and A-delivers~$m$.

\subsubsection{Non-faulty scenario}
We split the A-broadcast of $m$ in two: (1) \rbroadcast{m}; and (2) the empty messages $m_\emptyset$ travel back to \sender{m}. 
In a non-failure scenario, messages are R-broadcast in $D$ steps, i.e., $T_{D}(\mathit{msg}) = T(\mathit{msg})D$.
Moreover, to account for contention while sending to $d$ successors, we add to the sending overhead the expected waiting time,
i.e., $o_s = o + \frac{d-1}{2} o$. 
Note that for \rbroadcast{m}, there is no contention while receiving (every server, except \sender{m}, 
is idle until it receives $m$). Thus, the time to R-broadcast $m$ is estimated by $T_{D}(m)  = (L + o_s + o)D$.

\begin{figure}[!tp]
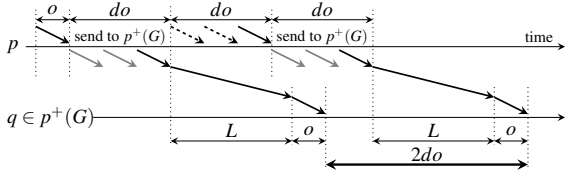

\centering
\ifdefined\PAPER
\scalebox{0.9}{\input{figures/recv_contention.tex}}
\fi
\ifdefined\TECHREP
\scalebox{0.85}{\input{figures/recv_contention.tex}}
\fi
\caption{LogP model of message transmission in \LSC{} for $d=3$. 
Dashed arrows indicate already sent messages.}
\label{fig:msg_logp}
\end{figure}

When the empty messages $m_\emptyset$ are transmitted to \sender{m}, the servers are no longer idle; 
$T_{D}(m_\emptyset)$ needs to account for contention while receiving. 
On average, servers send further one in every $d$ received messages; 
thus, a server $p$ sends messages to the same successor $q$ at a period of $2do$ (see Figure~\ref{fig:msg_logp}). 
In general, once a message arrives at a server, it needs to contend with other received messages. 
Yet, servers handle incoming connections in a round-robin fashion; 
processing a round of messages from all $d$ predecessors takes (on average) $2do$, 
i.e., $2o$ per predecessor  (see server $p$ in Figure~\ref{fig:msg_logp}).
Thus, on average, the message in-rate on a connection matches the out-rate: 
There is no contention while receiving empty messages, i.e., $T_{D}(m_\emptyset) = T_{D}(m)$.

\subsubsection{Faulty scenario---probabilistic analysis}
\label{sec:prob_analysis}

Let $\pi_\mathit{m}$ be the longest path a message $m$ has to travel before it is completely disseminated. 
If $m$ is lost (due to failures), $\pi_\mathit{m}$ is augmented by the longest path the failure notifications 
have to travel before reaching all non-faulty servers.
Let $\mathcal{D}$ be a random variable that denotes the length of the longest path $\pi_\mathit{m}$, 
for any A-broadcast message $m$, i.e., $\mathcal{D} = \max_m |\pi_\mathit{m}|,\,\forall m$; 
we refer to $\mathcal{D}$ as \LSC{}'s \emph{depth}.
Intuitively, the depth is the asynchronous equivalent of the number of rounds from synchronous systems.
Thus, $\mathcal{D}$ ranges from $D$, if no servers fail, to $f + D_f$ in the worst case scenario~(\cref{sec:lowerbound}). 
Yet, $\mathcal{D}$ is not uniformly distributed. A back-of-the-envelope calculation shows that it is very unlikely for \LSC{}'s depth to exceed 
$D_f$.

We consider a single \LSC{} round, with all $n$ servers initially non-faulty. 
Also, we estimate the probability $p_f$ of a server to fail, by using an exponential lifetime distribution model, 
i.e., over a period of time $\Delta$, $p_f=1-e^{-\Delta/\textit{MTTF}}$, where $\mathit{MTTF}$ is the mean time to failure.
If \sender{m} succeeds in sending $m$ to all of its successors, then $D\leq \pi_\mathit{m} \leq D_f$~(\cref{sec:lowerbound}).
Thus, $\mathit{Pr}[D\leq \mathcal{D} \leq D_f]=e^{-n\cdot d\cdot o/\mathit{MTTF}}$, 
where $o$ is the sending overhead~\cite{Culler:1993:LTR:173284.155333}.
Note that this probability increases if the round starts with previously failed servers. 

For typical values of $\mathit{MTTF}$ ($\approx2$~years~\cite{Sato:2012:DMN:2388996.2389022})
and $o$ ($\approx1.8\mu s$ for TCP on our InfiniBand cluster~\cref{sec:lsc_eval}), a system of $256$ servers connected 
via a digraph of degree $7$ 
(see Table~\ref{tab:sim_digraph}) 
would finish $1$ million \LSC{} rounds with 
$\mathcal{D} \leq D_f$ with a probability larger than $99.99\%$.
This demonstrates why early termination is essential for efficiency, as for most rounds no failures occur 
and even if they do occur, the probability of $\mathcal{D} > D_f$ is very small. 
Note that a practical deployment of \LSC{} should include regularly replacing failed servers and/or updating $G$ after 
failures.

\subsubsection{Estimating the fault diameter}
\label{sec:estimdf}

The fault diameter of any digraph $G$ is trivially bounded by 
$\left\lfloor \frac{n - f - 2} {k(G)-f} \right\rfloor + 1$~\cite[Theorem~6]{Chung1984}. 
However, this bound is neither tight nor does it relate the fault diameter to the digraph's diameter. 
In general, the fault diameter is unbounded in terms of the digraph diameter~\cite{Chung1984}. 
Yet, if the first $f+1$ shortest vertex-disjoint paths from $u$ to $v$ are of length 
at most $\delta_f$ for $\forall u,v\in V(G)$, then $D_f(G,f) \leq \delta_f$~\cite{Krishnamoorthy:1987:FDI:35064.36256}.
To compute $\delta_f$, we need to solve the min-max $(f+1)$-disjoint paths problem for every pair of vertices: 
Find $(f+1)$ vertex-disjoint paths $\pi_0,\ldots,\pi_f$ that minimize the length of the longest path; 
hence, $\delta_f=\max_i{|\pi_i|},\,0\leq i \leq f$.

Unfortunately, the problem is known to be strongly NP-com\-plete~\cite{Li:1989:CFT:85167.85178}.
As a heuristic to find $\delta_f$, we minimize the sum of the lengths 
instead of the maximum length, i.e., the min-sum disjoint paths problem.
This problem can be expressed as a minimum-cost flow problem; 
thus, it can be solved polynomially with well known algorithms, 
e.g., successive shortest path~\cite[Chapter 9]{Ahuja:1993:NFT:137406}.
Let $\hat{\pi}_0,\ldots,\hat{\pi}_f$ be the paths obtained from solving the min-sum disjoint 
paths problem; also, let $\hat{\delta_f}=\max_i{|\hat{\pi}_i|},\,0\leq i \leq f$.
Then, from the minimality condition of both min-max and min-sum problems, we deduce the following 
chain of inequalities:
\begin{equation} \label{eq:df_bound}
\frac{\sum\limits_{i=0}^f{|\hat{\pi}_i|}}{f+1} 
\leq \frac{\sum\limits_{i=0}^f{|\pi_i|}}{f+1} 
\leq \delta_f
\leq \hat{\delta_f}.
\end{equation}
Thus, we approximate the fault diameter bound by $\hat{\delta_f}$.
Then, we use Equation~\eqref{eq:df_bound} to check the accuracy of our approximation: 
We check the difference between the maximum and the average length of the paths obtained 
from solving the tractable min-sum problem.

As an example, we consider the binomial graph example from~\cite{Angskun2007}, 
i.e., $n=12$ and $p_i^+=p_i^-=\left\{p_j:j=i\pm\{1,2,4\}\right\}$. The graph has connectivity $k=6$ and 
diameter $D=2$. After solving the min-sum problem, we can estimate the fault diameter bound, 
i.e., $3 \leq \delta_f \leq 4$. 
After a closer look, we can see that one of the six vertex-disjoint paths from $p_0$ 
to $p_3$ has length four, i.e., $p_0 - p_{10} - p_{6} - p_{5} - p_{3}$.

\subsection{Storage requirements}

\begin{table}[!tp]
\centering
\ifdefined\PAPER
\small
\fi
\ifdefined\TECHREP
\scriptsize
\fi
\begin{tabular}{ l  l l }

  \cmidrule[1.5pt](){1-3}

  \textbf{Notation} &
  \textbf{Description} &
  \textbf{Space complexity per server} \\
  \hline

  \rowcolor{DarkGray}  
  $G$ & 
  digraph & 
  $\mathcal{O}(n \cdot d)$ \\  
   
  $M_i$ & 
  messages &
  $\mathcal{O}(n)$\\  
   
   \rowcolor{DarkGray}  
   $F_i$ & 
   failure notifications &
   $\mathcal{O}(f \cdot d)$ \\  

   $\mathbf{g_i}$& 
   tracking digraphs &
   $\mathcal{O}(f^2 \cdot d)$ \\
  
   \rowcolor{DarkGray}  
   $Q$ &
   FIFO queue & 
   $\mathcal{O}(f \cdot d)$\\
    
  \cmidrule[1.5pt](){1-3} 
\end{tabular}
  \caption{Space complexity per server for data structures used by Algorithm~\ref{alg:lsc}. The space complexity for $G$ holds for regular digraphs, such as $G_S(n,d)$~\cref{sec:choose_g}.
}
\label{tab:space_complex}
\end{table}

Each server $p_i$ stores five data structures (see Algorithm~\ref{alg:lsc}): 
(1) the digraph $G$; (2) the set of known messages $M_i$; (3) the set of received failure notifications $F_i$; 
(4) the array of tracking digraphs $\mathbf{g_i}$; and (5) the internal FIFO queue $Q$. 
Table~\ref{tab:space_complex} shows the space complexity of each data structure. 
In general, for regular digraphs, $p_i$ needs to store $d$ edges per node; 
yet, some digraphs require less storage, e.g., binomial graphs~\cite{Angskun2007} require only the graph size.
Also, each tracking digraph has at most $fd$ vertices; yet, only $f$ of these digraphs may have more than one vertex.
The space complexity of the other data structures is straightforward (see Table~\ref{tab:space_complex}).

\subsection{Choosing the digraph $G$}
\label{sec:choose_g}

\LSC{}'s performance depends on the parameters of $G$---degree, diameter, and fault diameter.
Binomial graphs have both diameter and fault diameter lower than other commonly used graphs, 
such as the binary Hypercube~\cite{Angskun2007}. 
Also, they are optimally connected, hence, offering optimal work for the provided connectivity. 
Yet, their connectivity depends on the number of vertices, which reduces their flexibility: 
Binomial graphs provide either too much or not enough connectivity. 

\begin{figure}[!tp]
\centering
\includegraphics[width=.3\textwidth]{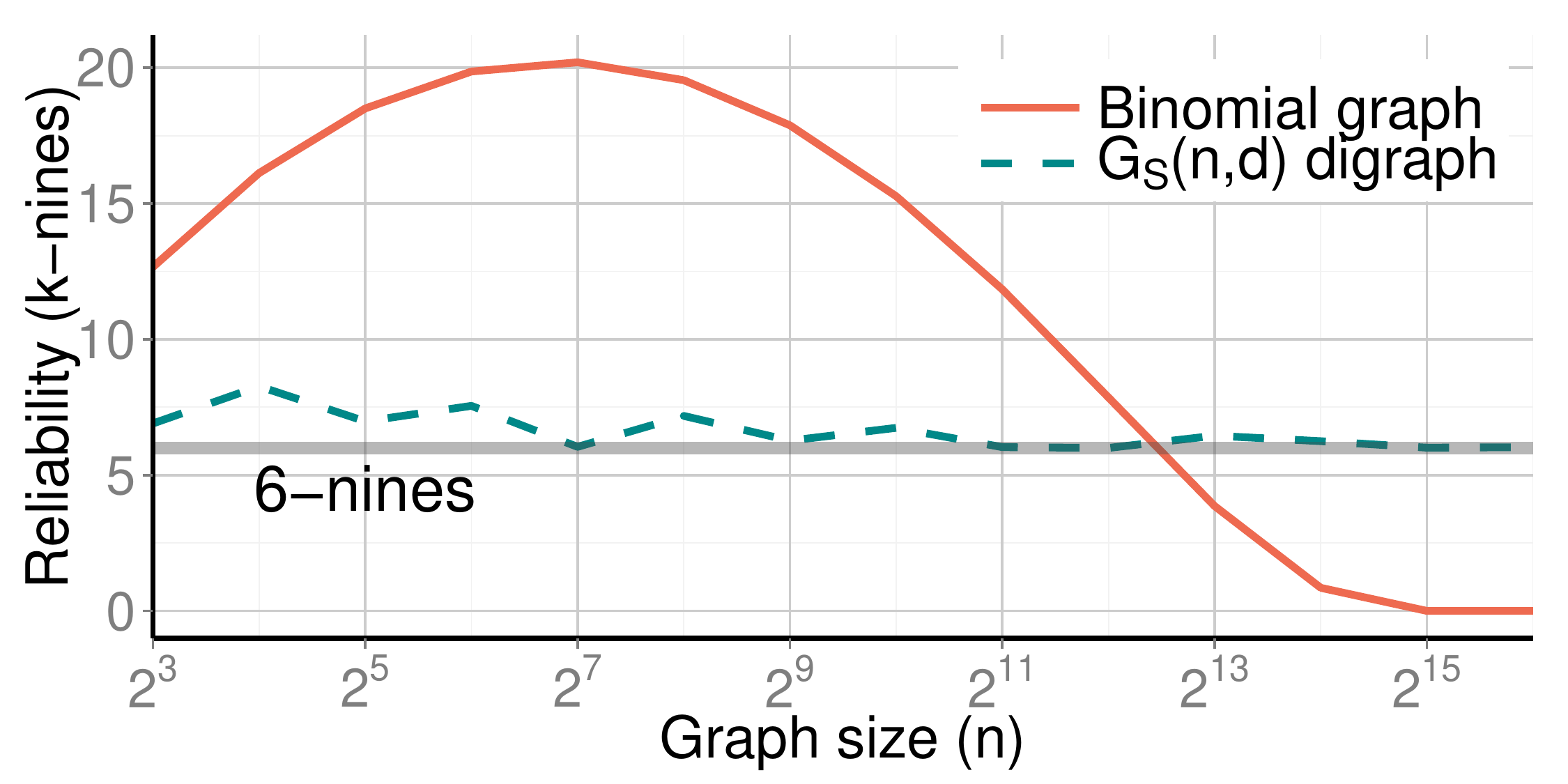}
\caption{\LSC{}'s reliability estimated over a period of $24$ hours and a server $\mathit{MTTF}\approx2$ years.}
\label{fig:digraph_reliability}
\end{figure}

We estimate \LSC{}'s reliability by 
$\rho_G = \sum_{i=0}^{k(G)-1}C(n,i)\cdot p^i_f(1-p_f)^{n-i}$,
with $p_f=1-e^{-\frac{\Delta}{\textit{MTTF}}}$ the probability of a server to fail over a period of time $\Delta$~(\cref{sec:prob_analysis}).
Figure~\ref{fig:digraph_reliability} plots this reliability as a function of $n$.
For a reliability target of 6-nines, we can see that the binomial graph offers either too much reliability, 
resulting in unnecessary work, or not sufficient reliability. 

As an alternative, \LSC{} uses $G_S(n,d)$ digraphs, for any $d\geq3$ and $n\geq 2d$~\cite{Soneoka:1996:DDC:227095.227101}.
In a nutshell, the construction of $G_S(n,d)$ entails constructing the line digraph of a 
generalized de Bruijn digraph~\cite{Du:1988:GDB:58354.58358} with the self-loops replaced by cycles.
\ifdefined\PAPER
A more detailed description that follows the steps provided in the original paper~\cite{Soneoka:1996:DDC:227095.227101} 
is available in an extendend technical report~\cite{poke2016allconcur_tr}.
\fi
\ifdefined\TECHREP
A more detailed description follows the steps provided in the original paper~\cite{Soneoka:1996:DDC:227095.227101}.

\textbf{Construction.}
Let $m$ be the quotient and $t$ the remainder of the division of $n$ by $d$, i.e., $n=md+t,\,m\geq2$.
Let $G_B(m,d)$ be a generalized de Bruijn digraph with $m$ vertices and degree $d$, 
i.e., $V(G_B(m,d))=\{0,\ldots,m-1\}$ and $E(G_B(m,d))=\{(u,v):v=ud+a(\text{mod}\, m),\,a=0,\ldots,d-1\}$.
Every vertex of $G_B(m,d)$ has at least $\lfloor d/m \rfloor$ self-loops; moreover, at least two vertices (i.e., $0$ and $m-1$), 
have $\lceil d/m \rceil$ self-loops. 
Therefore, we can remove the self-loops and replaced them with $\lfloor d/m \rfloor$ cycles connecting all the vertices 
and an additional cycle connecting only the vertices with $\lceil d/m \rceil$ self-loops; we denote the resulting $d$-regular digraph by $G_B^\ast(m,d)$.

Further, we construct the line digraph of $G_B^\ast(m,d)$, 
$L(G_B^\ast(m,d))$, which has a set of $md$ vertices $V^\prime=\{\mathit{uv}:(u,v)\in E(G_B^\ast(m,d))\}$ 
and a set of directed edges $E^\prime=\{(\mathit{uv},\mathit{wz}):v=w\}$.
If $t=0$, then $L(G_B^\ast(m,d))$ is the $G_S(n,d)$ digraph.
If $t>0$, then we choose an arbitrary vertex $v\in V(G_B^\ast(m,d))$.
Let $X=\{x_0,\ldots,x_{d-1}\}$ be a subset of $V^\prime$ with $d$ vertices $\mathit{uv},\,\forall u\in V(G_B^\ast(m,d))$;
similarly, let $Y=\{y_0,\ldots,y_{d-1}\}$ be a subset of $V^\prime$ with $d$ vertices $\mathit{vu},\,\forall u\in V(G_B^\ast(m,d))$.
Clearly, $X$ and $Y$ exist since $G_B^\ast(m,d)$ is $d$-regular.
Moreover, let $M=\{(x,y):\forall x\in X,\,\forall y\in Y\}$; clearly, $M$ is a subset of $E^\prime$.
Then, $G_S(n,d)$ is constructed by adding a set of $t$ vertices, i.e., $W=\{w_0,\ldots,w_{t-1}\}$, to $L(G_B^\ast(m,d))$ as follows:
\begin{align*}
 V(G_S(n,d))&=V^\prime \cup W\\
 E(G_S(n,d))&=E^\prime \cup \{(w_i,w_j):i\neq j\}\\
	    &\bigcup_{i=0}^{t-1}\{(x,w_i),(w_i,y):x\in X_i,y\in Y_i\} \setminus \bigcup_{i=0}^{t-1}M_i,
\end{align*}
where $M_i=\{(x_{i+p},y_{i+q}):q=i+p(\text{mod } d-t+1),\,0\leq p \leq d-t\}$,
 $X_i=\{x_i,\ldots,x_{i+d-t}\}$, and $Y_i=\{y_i,\ldots,y_{i+d-t}\}$,
for $i=0,\ldots,t-1$.

\textbf{Properties.}
\fi
Similarly to binomial graphs~\cite{Angskun2007}, $G_S(n,d)$ digraphs are optimally connected. 
Contrary to binomial graphs though, they can be adapted to various reliability targets 
(see Figure~\ref{fig:digraph_reliability} for a reliability target of 6-nines).
Moreover, $G_S(n,d)$ digraphs have a quasiminimal diameter for $n \leq d^3+d$:
The diameter is at most one larger than the lower bound obtained from the Moore bound, 
i.e., $\mathit{DL}(n,d) = \lceil\log_d{(n(d-1)+d)}\rceil-1$. 
In addition, $G_S(n,d)$ digraphs have low fault diameter bounds (experimentally verified).
Table~\ref{tab:sim_digraph} shows the parameters of $G_S(n,d)$ for different number of vertices and 
6-nines reliability; the reliability is estimated over a period of $24$ hours according to the data 
from the TSUBAME2.5 system failure history~\cite{Sato:2012:DMN:2388996.2389022,tsubame}, 
i.e., server $\mathit{MTTF}\approx2$ years.

\begin{table}[!tp]
\centering
\ifdefined\PAPER
\small
\fi
\ifdefined\TECHREP
\scriptsize
\fi
\begin{tabular}{ l c c | l c c}
  \cmidrule[1.5pt](){1-6}

  $G_S(n,d)$ & $D$ & $\mathit{DL}(n,d)$ & $G_S(n,d)$ & $D$ & $\mathit{DL}(n,d)$ \\ 
  \hline

  \rowcolor{DarkGray}  
  $G_S(6,3)$ & 2 & 2 & $G_S(64,5)$ & 4 & 3 \\ 
  
  $G_S(8,3)$ & 2 & 2 & $G_S(90,5)$ & 3 & 3 \\ 
  
  \rowcolor{DarkGray}
  $G_S(11, 3)$ & 3 & 2 & $G_S(128,5)$ & 4 & 3 \\ 
  
  $G_S(16,4)$ & 2 & 2 & $G_S(256,7)$ & 4 & 3 \\ 
  
  \rowcolor{DarkGray}
  $G_S(22,4)$ & 3 & 3 & $G_S(512,8)$ & 3 & 3 \\ 
  
  $G_S(32,4)$ & 3 & 3 & $G_S(1024,11)$ & 4 & 3 \\ 
  
  \rowcolor{DarkGray}
  $G_S(45,4)$ & 4 & 3 &  &  &  \\ 
  
  \cmidrule[1.5pt](){1-6} 
\end{tabular}
  \caption{The parameters---vertex count $n$, degree $d$ and diameter~$D$---of $G_S(n,d)$ for 6-nines reliability (estimated over a period of $24$ hours and a server $\mathit{MTTF}\approx2$ years).
  The lower bound for the diameter is $\mathit{DL}(n,d)=\lceil\log_d{(n(d-1)+d)}\rceil-1$.}
\label{tab:sim_digraph}
\end{table}

\subsection{AllConcur vs. leader-based agreement}
\label{sec:comparison}

For a theoretical comparison to leader-based agreement, we consider the following deployment: 
a leader-based group, such as Paxos, that enables the agreement among $n$ servers, i.e., clients in Paxos terminology (see Figure~\ref{fig:paxos}). 
The group size does not depend on $n$, but only on the reliability of the group members.
Also, all the servers interact directly with the leader. 
In principle, the leader can disseminate state updates via a tree~\cite{hoefler-moor-collectives};
yet, for fault-tolerance, a reliable broadcast algorithm~\cite{Buntinas:2011:SDC:2042476.2042515} 
is needed. To the best of our knowledge, there is no implementation of leader-based agreement that uses reliable broadcast 
for dissemination.

In general, in such a leader-based deployment, not all servers need to send a message. 
This is an advantage over \LSC{}, where the early termination mechanism requires every server to send a message. 
Yet, for the typical scenarios targeted by \LSC{}---the data to be agreed upon is well balanced---we can 
realistically assume that all servers have a message to send. 

\textbf{Trade-off between work and total message count.} 
The work require for reaching agreement in a leader-based deployment is unbalanced.
On the one hand, every server sends one message and receives $n-1$ messages, resulting in $O(n)$ work.
On the other hand, the leader requires quadratic work, i.e., $\mathcal{O}(n^2)$: 
it receives one message from every server and it sends every received message to all servers.
Note that every message is also replicated, adding a constant amount of work per message.

To avoid overloading a single server (i.e., the leader), \LSC{} distributes the work evenly 
among all servers---every server performs $\mathcal{O}(nd)$ work~(\cref{sec:work}).
This decrease in complexity comes at the cost of introducing more messages to the network. 
A leader-based deployment introduces $n(n-1)$ messages to the network 
(not counting the messages needed for replication). In \LSC{}, every message is sent $d$ times; 
thus, the total number of messages in the network is $n^2d$. 

\textbf{Removing and adding servers.} 
For both \LSC{} and leader-based agreement, the cost of intentionally removing and adding servers 
can be hidden by using a two-phase approach similar to the transitional configuration in Raft~\cite{Ongaro2014}.
Thus, we focus only on the cost of unintentionally removing a server---a server failure. 
Also, we consider a worst-case analysis---we compare the impact of a leader failure to that of a \LSC{} server. 
The consequence of a leader failure is threefold: 
(1) every server receives one failure notification; 
(2) a leader election is triggered; 
and (3) the new leader needs to reestablish the connections to the $n$ servers. 
Note that the cost of reestablishing the connection can be hidden if the servers 
connect from the start to all members of the group. 
In \LSC{}, there is no need for leader election. A server failure causes every server 
to receive up to $d^2$ failure notifications~(\cref{sec:work}).
Also, the depth may increase~(\cref{sec:prob_analysis}).

\textbf{Redundancy.} 
The amount of redundancy (i.e., $d$) needed by \LSC{} is given by the reliability of the agreeing servers.
Thus, $d$ can be seen as a performance penalty for requiring a certain level of reliability.
Using more reliable hardware increases \LSC{}'s performance. 
In contrast, in a leader-based deployment, 
more reliable hardware increases only the performance of message replication (i.e., less replicas are needed), 
leaving both the quadratic work and the quadratic total message count unchanged.

\section{Evaluation}
\label{sec:lsc_eval}

We evaluate \LSC{} on two production systems: 
(1) an InfiniBand cluster with 96 nodes; 
and (2) the Hazel Hen Cray XC40 system (7712 nodes).
We refer to the two systems as IB-hsw and XC40, respectively.
On both systems, each node has 128GB of physical memory and 
two Intel Xeon E5-2680v3 12-core CPUs with a base frequency of 2.5GHz.
The IB-hsw system nodes are connected through a Voltair 4036 Fabric (40Gbps);
each node uses a single Mellanox ConnectX-3 QDR adapter (40GBps). 
Moreover, each node is running ScientificLinux version 6.4.
The  XC40 system nodes are connected through the Cray Aries network.

We implemented \LSC{}\footnote{\code{}} in \textsc{C}; the implementation relies on 
\emph{libev}, a high-performance event loop library. Each instance of \LSC{} is deployed on a single physical node. 
The nodes communicate via either standard sockets-based TCP or high-performance InfiniBand Verbs (IBV); 
we refer to the two variants as \LSC{}-TCP and \LSC{}-IBV, respectively.
On the IB-hsw system, to take advantage of the high-performance network, we use TCP/IP over InfiniBand (``IP over IB'') for \LSC{}-TCP.
The failure detector is implemented over unreliable datagrams. 
To compile the code, we use GCC version 5.2.0 on the IB-hsw system and Cray Programming Environment 5.2.82 on the XC40 system.

We evaluate \LSC{} through a set of benchmarks that emulate representative real-world applications.
During the evaluation, we focus on two common performance metrics:
(1) the \emph{agreement latency}, i.e., the time needed to reach agreement; and 
(2) the \emph{agreement throughput}, i.e., the amount of data agreed upon per second.
In addition, we introduce the \emph{aggregated throughput}, a performance metric defined as the agreement throughput times the number of servers.
Also, all the experiments assume a perfect FD.

In the following benchmarks, the servers are interconnected via $G_S(n,d)$ digraphs (see Table~\ref{tab:sim_digraph}).
If not specified otherwise, each server generates requests at a certain rate.
The requests are buffered until the current agreement round is completed; then, they are packed into a message 
that is A-broadcast in the next round. 
All the figures report both the median and the $95\%$ nonparametric confidence interval around it~\cite{Hoefler:2015:SBP:2807591.2807644}. 
Moreover, for each figure, the system used to obtain the measurements is specified in square brackets.

\begin{figure}[!tp]
\centering
\subfloat[\label{fig:allconcur_ibv_latency} {\LSC{}-IBV [IB-hsw]} ] {
\includegraphics[width=.235\textwidth]{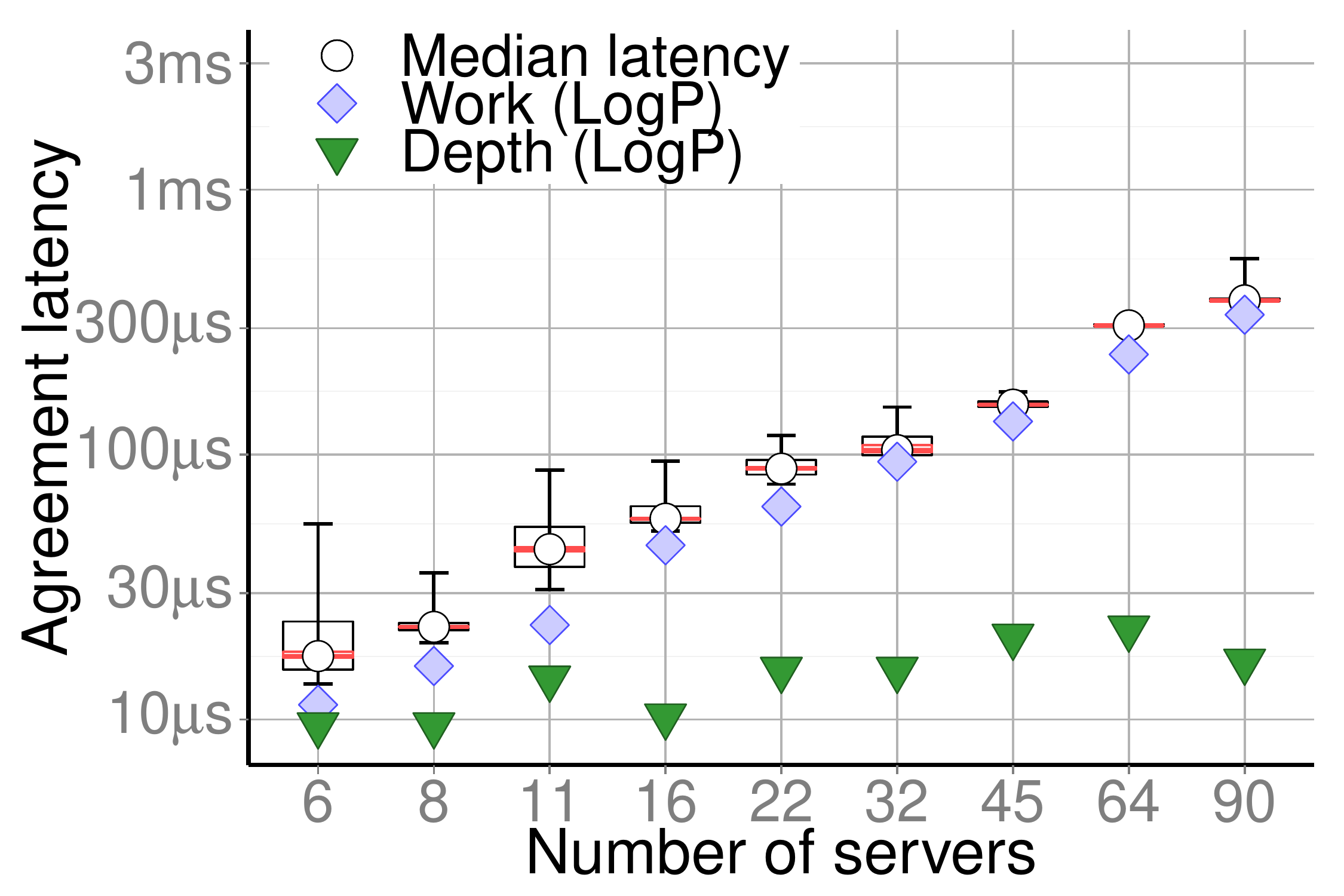}
}
\subfloat[\label{fig:allconcur_tcp_latency} {\LSC{}-TCP [IB-hsw]}]{        
\includegraphics[width=.235\textwidth]{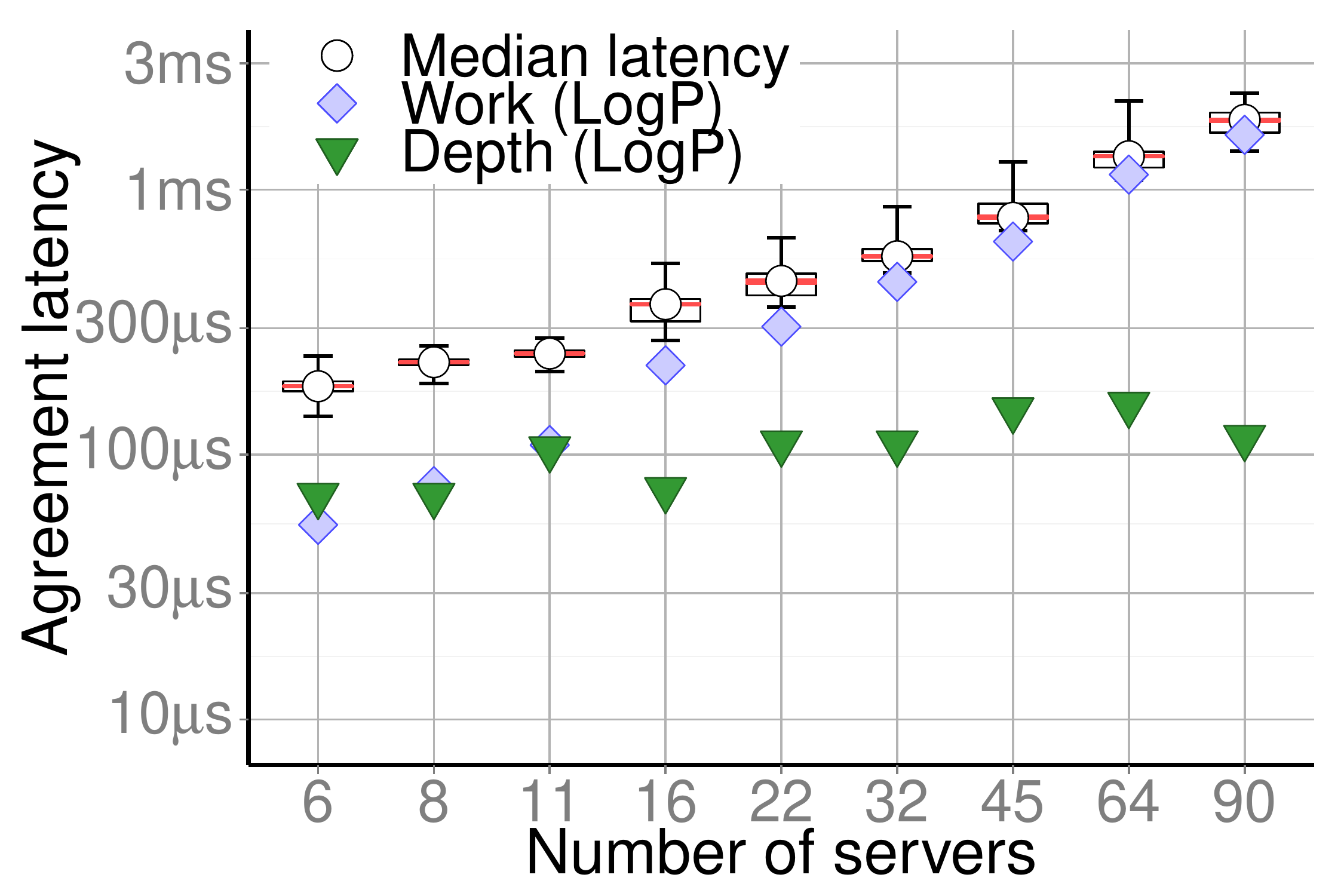}
}
\caption{Agreement latency for a single (64-byte) request. 
The LogP parameters are $L=1.25\mu s$ and $o=0.38\mu s$ over IBV 
and $L=12\mu s$ and $o=1.8\mu s$ over TCP.}
\label{fig:allconcur_latency}
\end{figure}

\textbf{Single request agreement.}
To evaluate the LogP models described in Section~\ref{sec:lsc_perform}, 
we consider a benchmark where the servers agree on one single request.
Clearly, such a scenario is not the intended use case of \LSC{}, as all servers, except one, 
A-broadcast empty messages. 
Figure~\ref{fig:allconcur_latency} plots the agreement latency as a function of system size 
for both \LSC{}-IBV and \LSC{}-TCP on the IB-hsw system and it compares it with the LogP models 
for both work and depth~(\cref{sec:lsc_perform}).
The LogP parameters for the IB-hsw system are $L=1.25\mu s$ and $o=0.38\mu s$ over IBV 
and $L=12\mu s$ and $o=1.8\mu s$ over TCP.
The models are good indicators of \LSC{}'s performance; e.g., with increasing the system size, work becomes dominant.

\begin{figure}[!tp]
\centering
\subfloat[\label{fig:allconcur_fail_all} {\LSC{}-IBV [IB-hsw]}] {
\includegraphics[width=.235\textwidth]{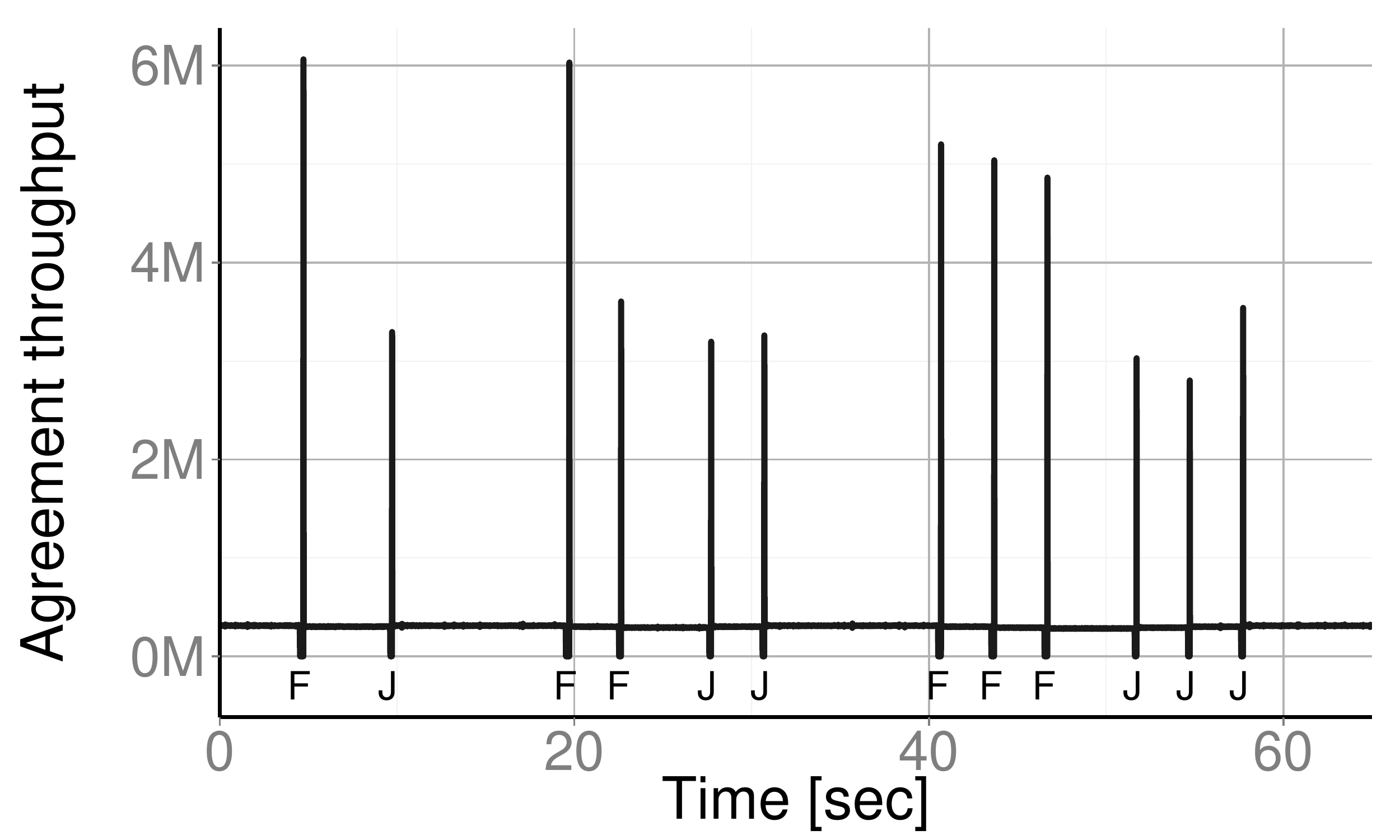}
}
\subfloat[\label{fig:allconcur_fail_zoom} Zoom-in of (a)]{        
\includegraphics[width=.235\textwidth]{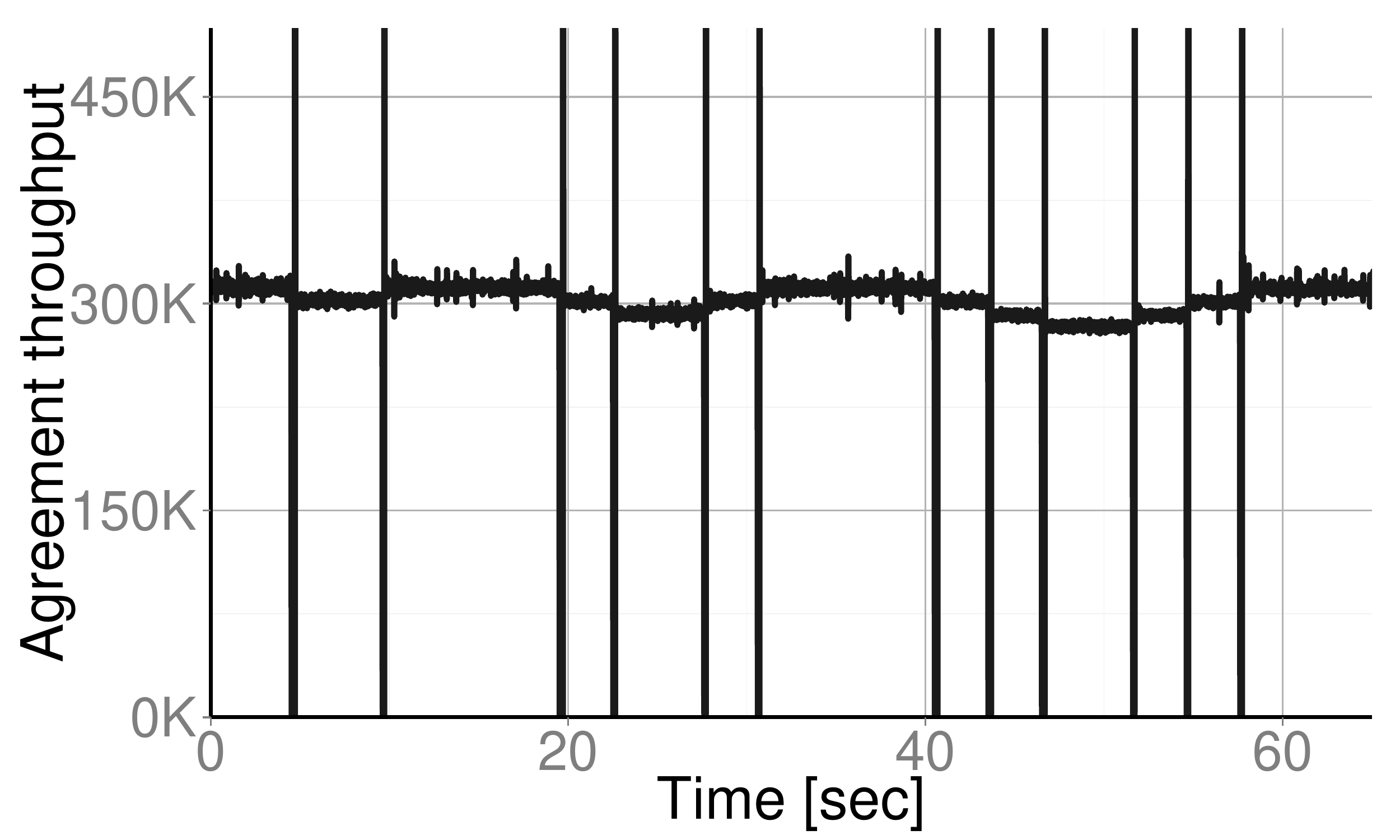}
}
\caption{Agreement throughput during membership changes---servers failing, indicated by F, and servers joining, indicated by J. 
Deployment over $32$ servers, each generating 10,000 (64-byte) requests per second.
The FD has $\Delta_{hb}=10ms$ and $\Delta_{to}=100ms$. 
The spikes in throughput are due to the accumulated requests during unavailability periods.
}
\label{fig:allconcur_fail}
\end{figure}

\ifdefined\PAPER
\begin{figure*}[!tp]
\centering
\subfloat[\label{fig:allconcur_ibv_rate_srv} {\LSC{}-IBV [IB-hsw]} ] {
\includegraphics[width=.24\textwidth]{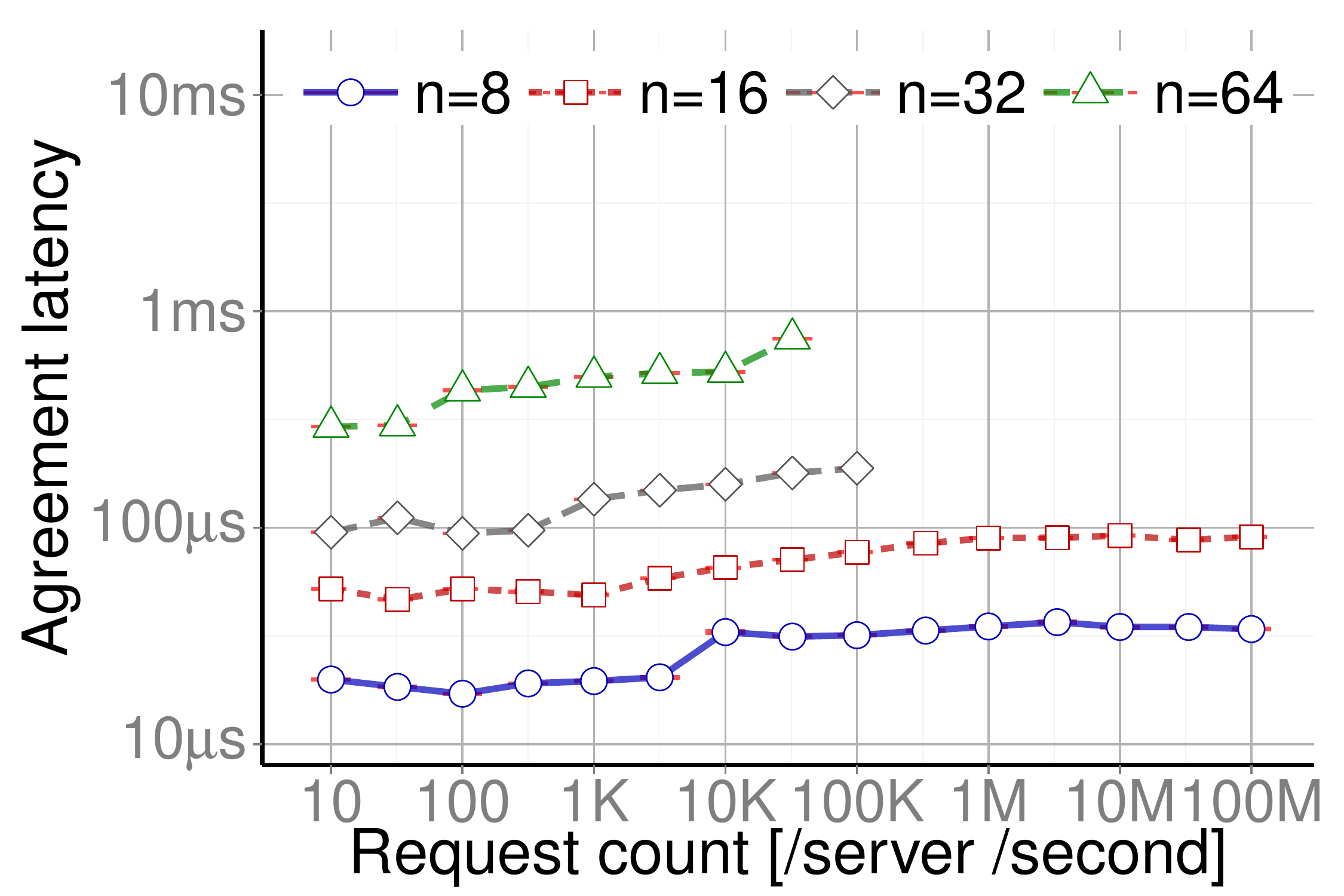}
}
~
\subfloat[\label{fig:allconcur_tcp_rate_srv} {\LSC{}-TCP [IB-hsw]} ] {
\includegraphics[width=.24\textwidth]{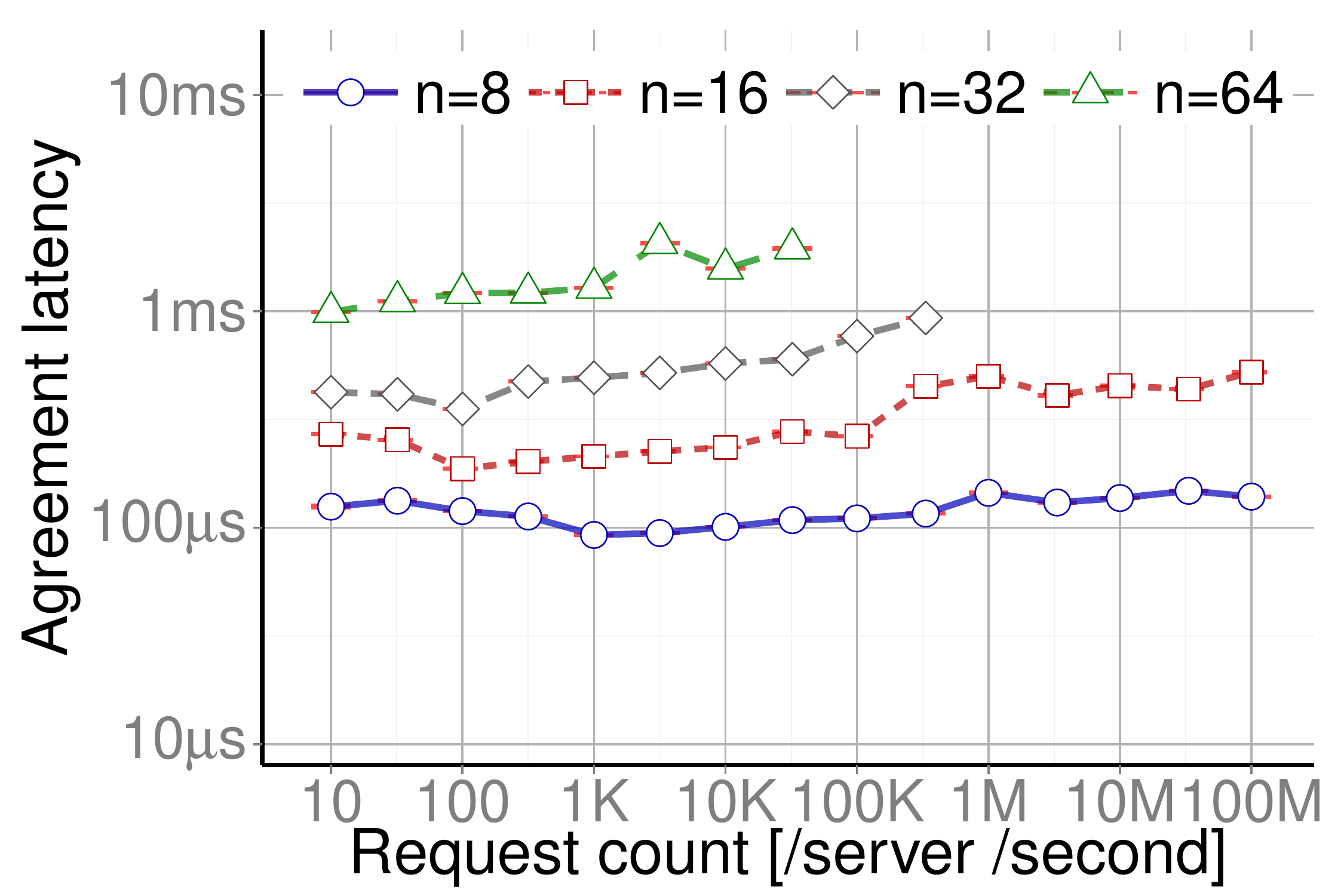}
}
~
\subfloat[\label{fig:allconcur_gaming} {\LSC{}-TCP [XC40]}] {
\includegraphics[width=.24\textwidth]{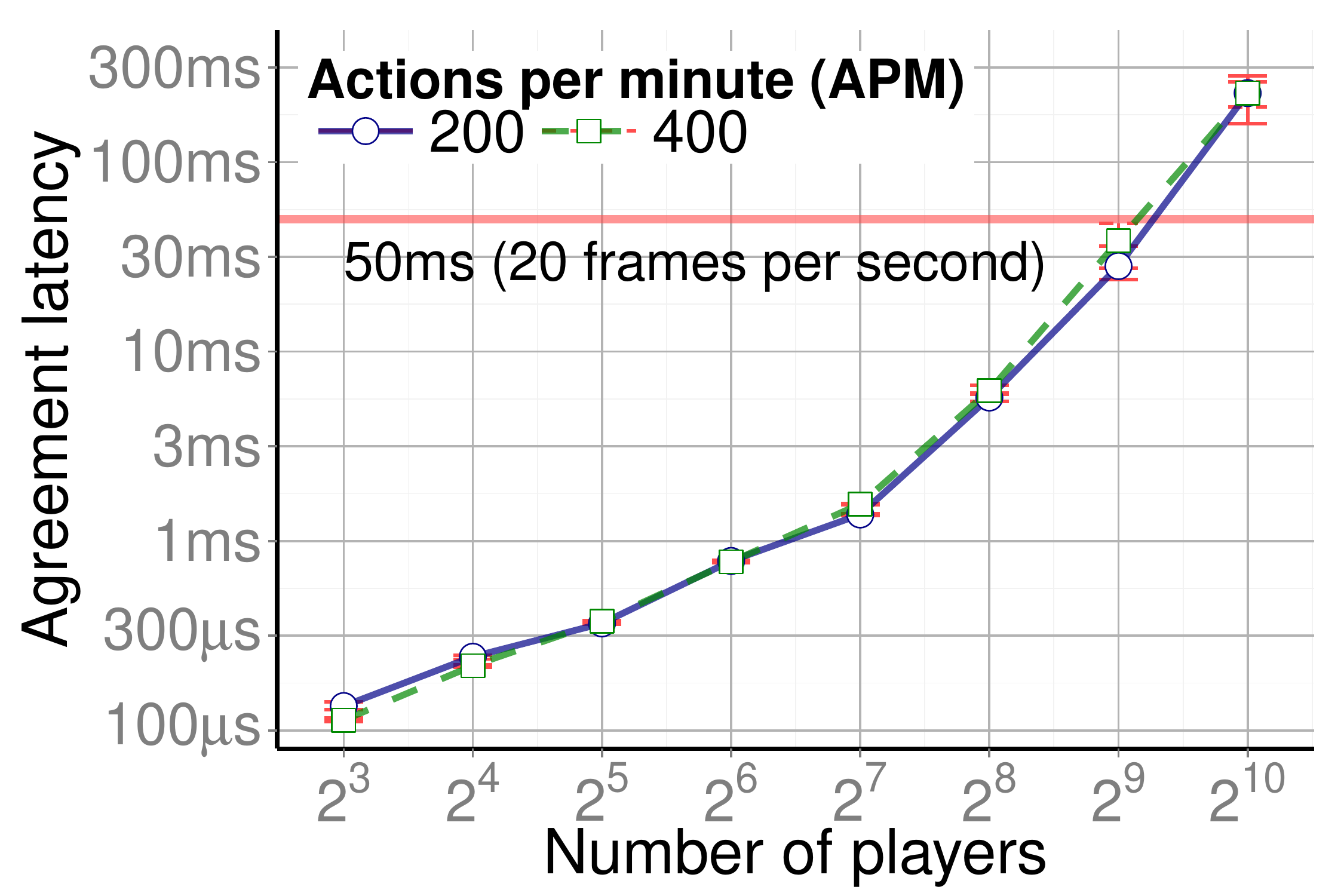}
}
\caption{(a),(b) Constant (64-byte) request rate per server.
(c) Agreement latency in multiplayer video games for different APM and 40-byte requests.}
\end{figure*}

\begin{figure*}[!tp]
\centering
\subfloat[\label{fig:mpi_allgather_throughput} {MPI\_Allgather [TCP / XC40]}] {
\includegraphics[width=.24\textwidth]{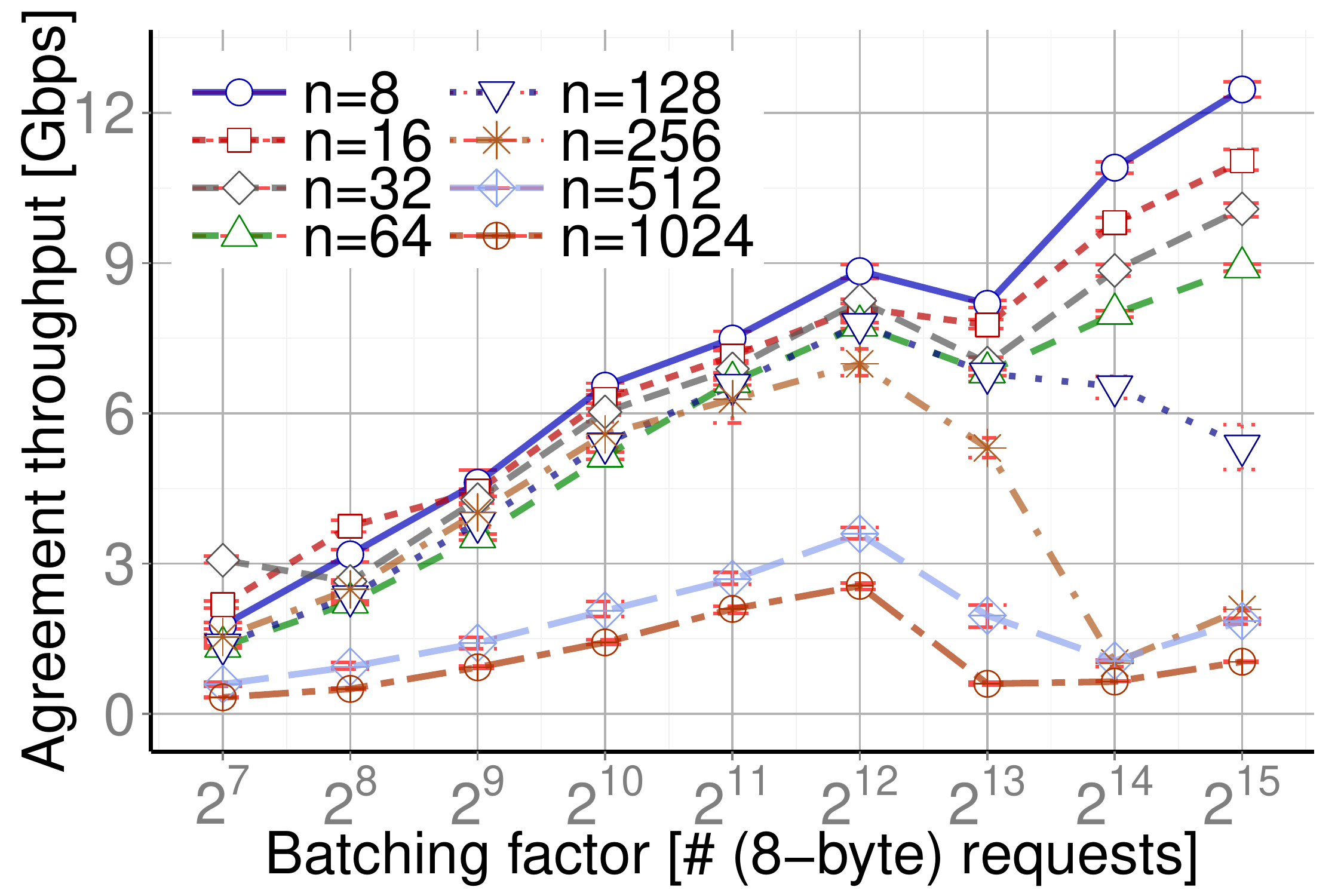}
}
~
\subfloat[\label{fig:allconcur_tcp_throughput} {\LSC{}-TCP [XC40]}] {
\includegraphics[width=.24\textwidth]{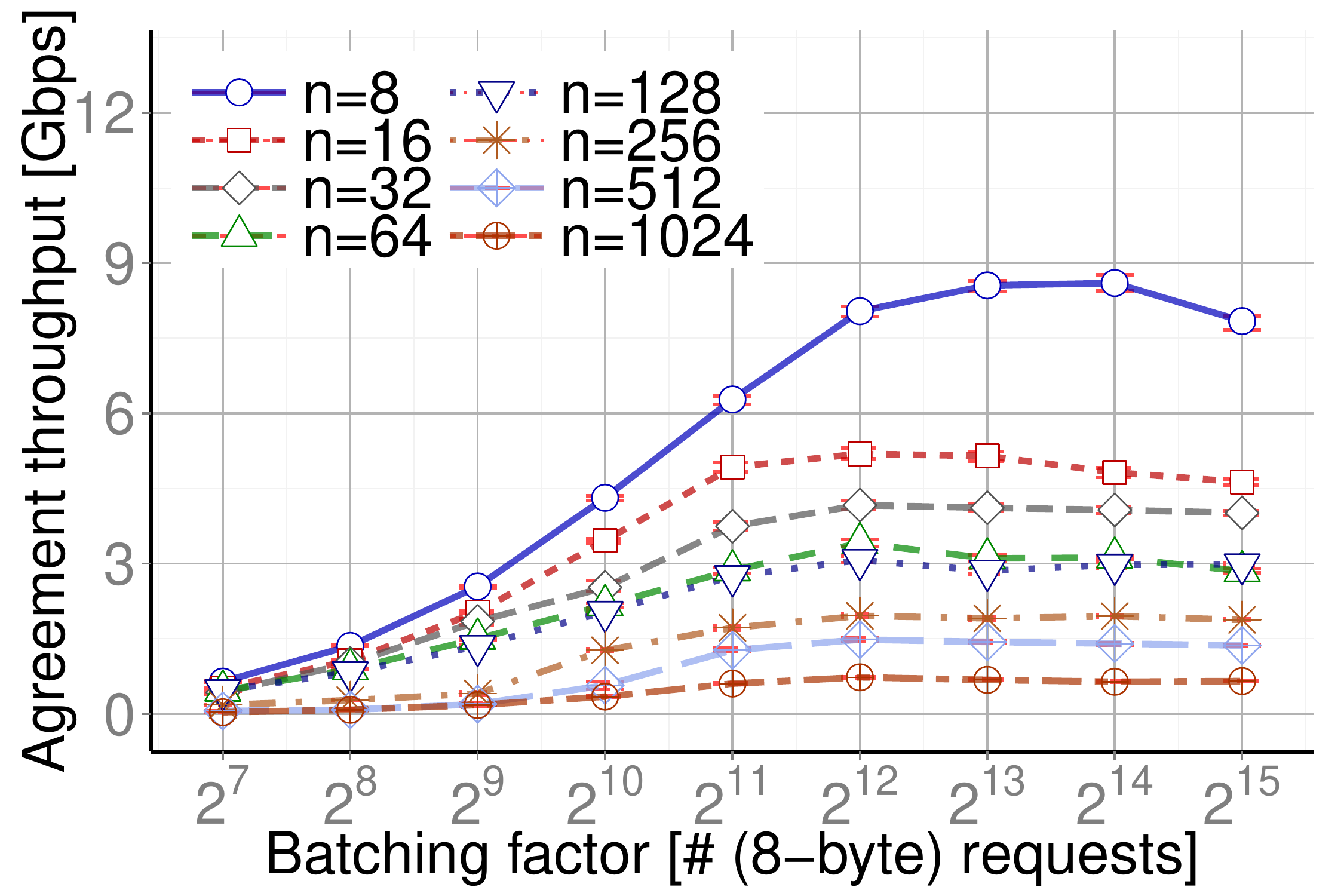}
}
~
\subfloat[\label{fig:libpaxos_throughput} {Libpaxos [TCP / XC40]}]{
\includegraphics[width=.24\textwidth]{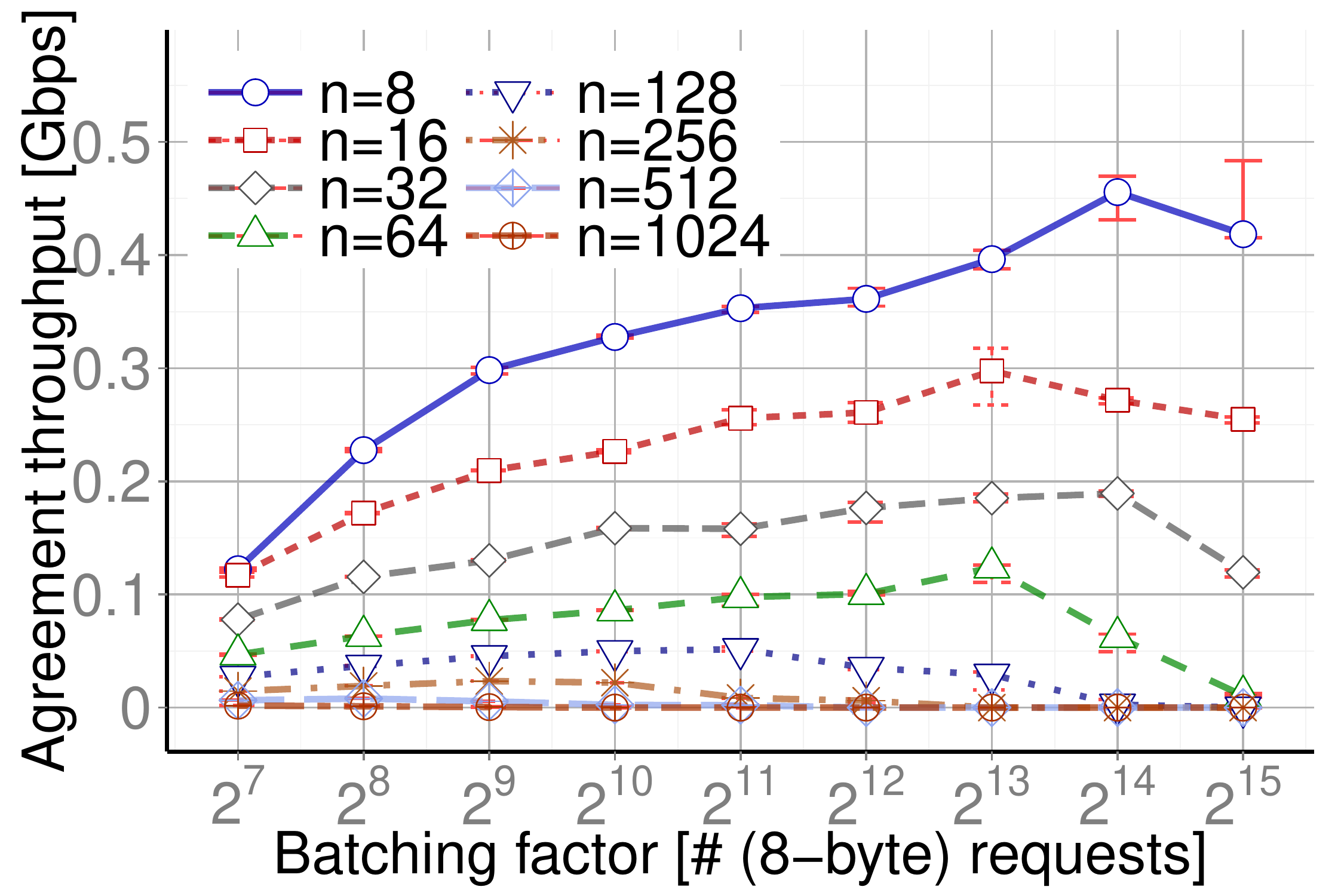}
}
~
\subfloat[\label{fig:allconcur_volume} {\LSC{}-TCP [XC40]}]{
\includegraphics[width=.24\textwidth]{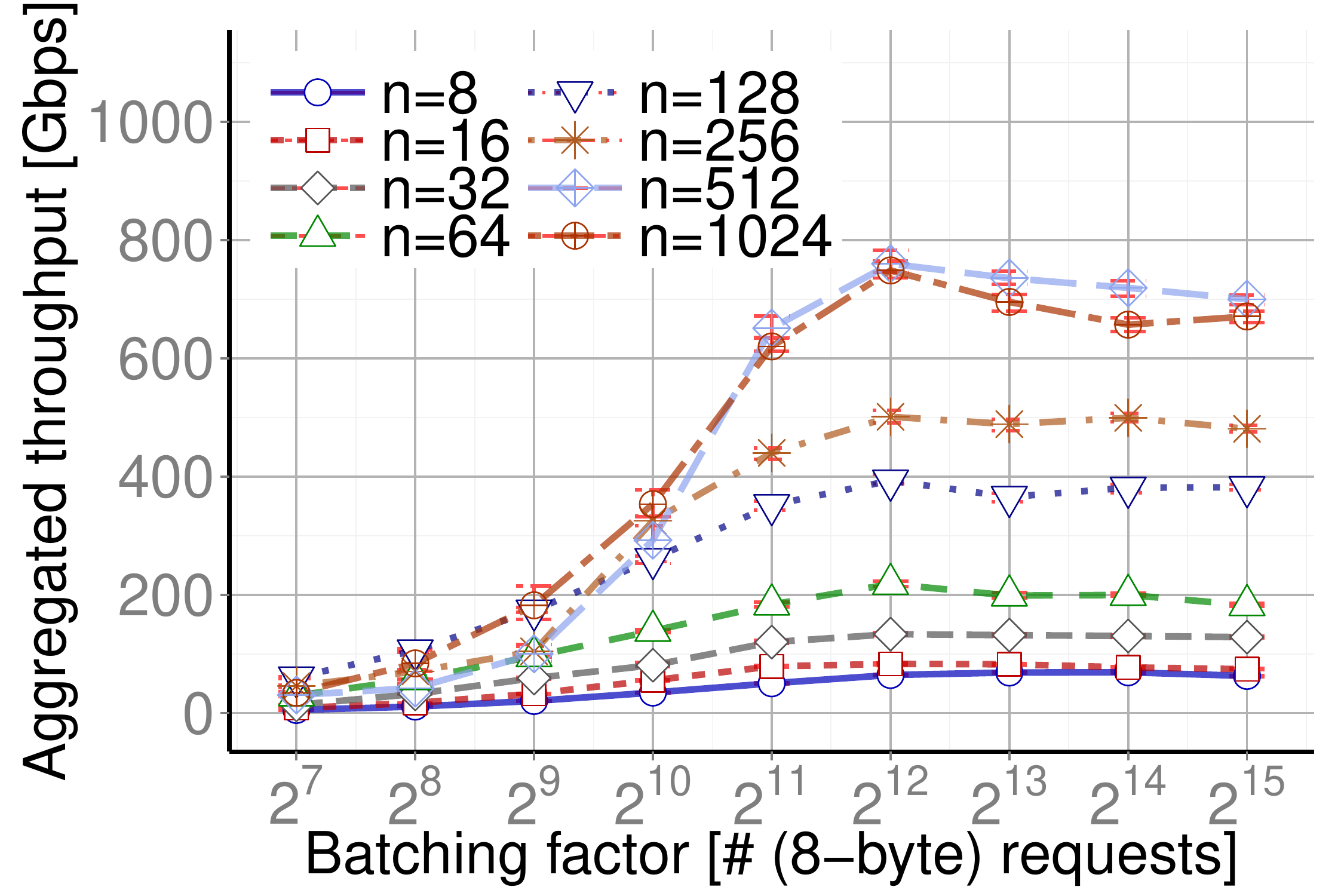}
}
\caption{(a) Unreliable agreement vs. (b) \LSC{} vs. (c) leader-based agreement---batching factor effect on the agreement throughput.
(d) Batching factor effect on the aggregated throughput.}
\label{fig:batching_factor}
\end{figure*}

\fi

\textbf{Membership changes.}
To evaluate the effect of membership changes on performance, we deploy \LSC{}-IBV on the IB-hsw system.
In particular, we consider 32 servers each generating 10,000 (64-byte) requests per second.
Servers rely on a heartbeat-based FD with a heartbeat period $\Delta_{hb}=10ms$ and a timeout period $\Delta_{to}=100ms$. 
Figure~\ref{fig:allconcur_fail} shows \LSC{}'s agreement throughput (binned into $10ms$ intervals) during a series of events, 
i.e., servers failing, indicated by F, and servers joining, indicated by J. 
Initially, one server fails, causing a period of unavailability ($\approx190ms$);
this is followed by a rise in throughput, due to the accumulated requests (see Figure~\ref{fig:allconcur_fail_all}). 
Shortly after, the system stabilizes, but at a lower throughput since one server is missing. 
Next, a server joins the system causing another period of unavailability $(\approx80ms)$ followed by another rise in throughput. 
Similarly, this scenario repeats for two and three subsequent 
failures\footnote{The $G_S(32,4)$ has vertex-connectivity four; thus, in general, it cannot safely sustain more than three failures}. 
Note that both unavailability periods can be reduced. 
First, by improving the FD implementation, $\Delta_{to}$ can be significantly decreased~\cite{Dragojevic:2015:NCD:2815400.2815425}.
Second, new servers can join the system as non-participating members until they established all necessary connections~\cite{Ongaro2014}.

\ifdefined\TECHREP
\begin{figure}[!tp]
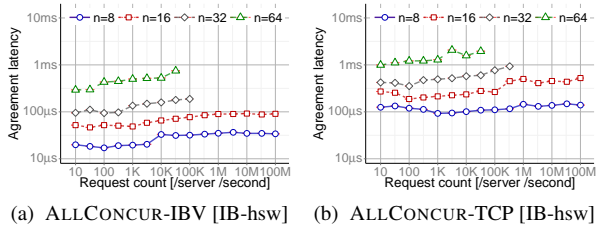

\centering
\subfloat[\label{fig:allconcur_ibv_rate_srv} {\LSC{}-IBV [IB-hsw]} ] {
\includegraphics[width=.235\textwidth]{figures/allconcur_ibv_rate_srv_round}
}
\subfloat[\label{fig:allconcur_tcp_rate_srv} {\LSC{}-TCP [IB-hsw]} ] {
\includegraphics[width=.235\textwidth]{figures/allconcur_tcp_rate_srv_round}
}
\caption{Constant (64-byte) request rate per server.}
\label{fig:allconcur_rate_srv}
\end{figure}
\fi

\ifdefined\TECHREP
\begin{figure}[!tp]
\centering
\subfloat[\label{fig:allconcur_gaming} {\LSC{}-TCP [XC40]}] {
\includegraphics[width=.235\textwidth]{figures/allconcur_gaming}
}
\subfloat[\label{fig:allconcur_rate_sys} {\LSC{}-TCP [XC40]}] {
\includegraphics[width=.235\textwidth]{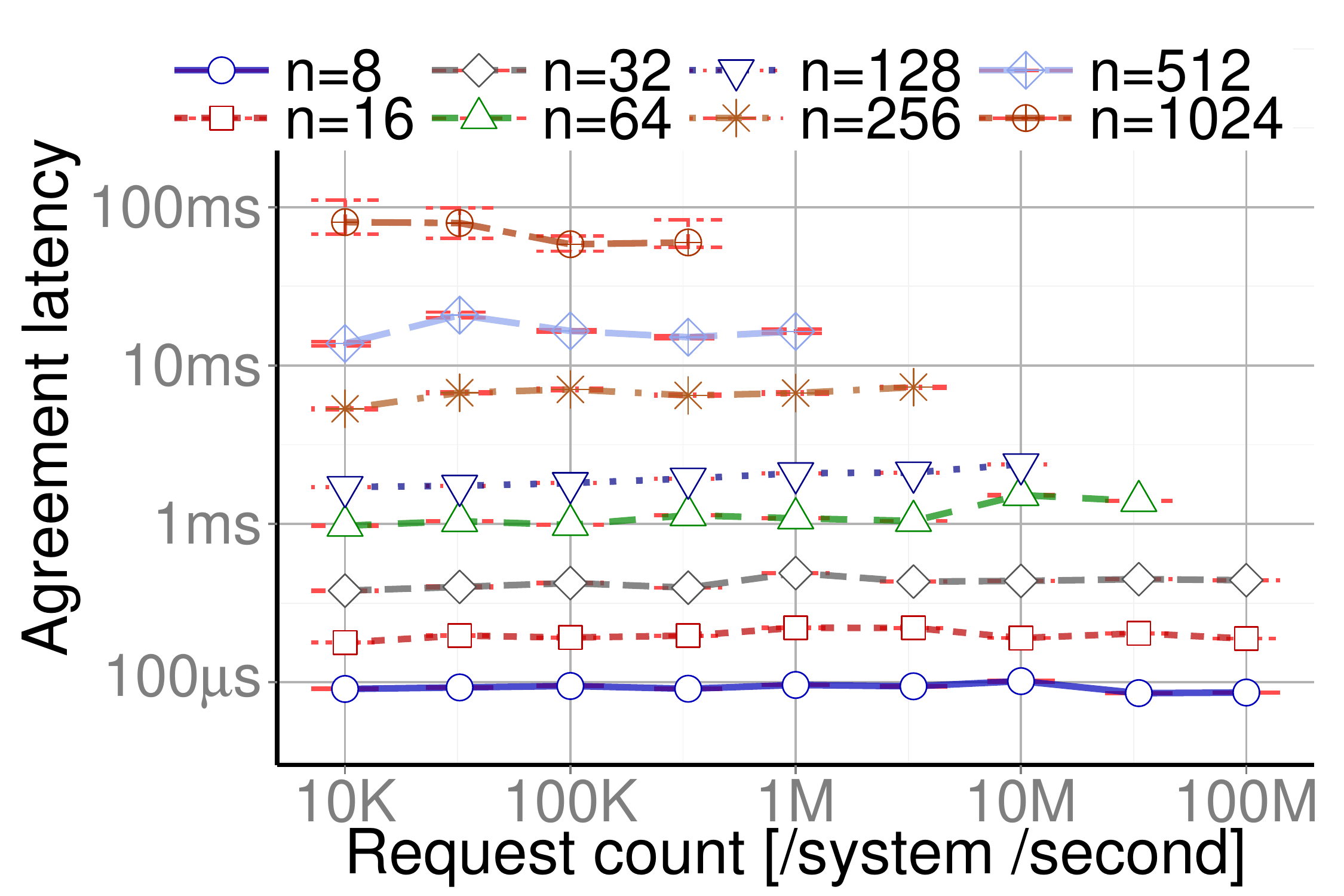}
}
\caption{(a) Agreement latency in multiplayer video games for different APM and 40-byte requests.
(b) Constant (40-byte) request rate per system.}
\label{fig:allconcur_game_sys}
\end{figure}
\fi

\ifdefined\TECHREP
\begin{figure*}[!tp]
\centering
\subfloat[\label{fig:mpi_allgather_throughput} {MPI\_Allgather [TCP / XC40]}] {
\includegraphics[width=.235\textwidth]{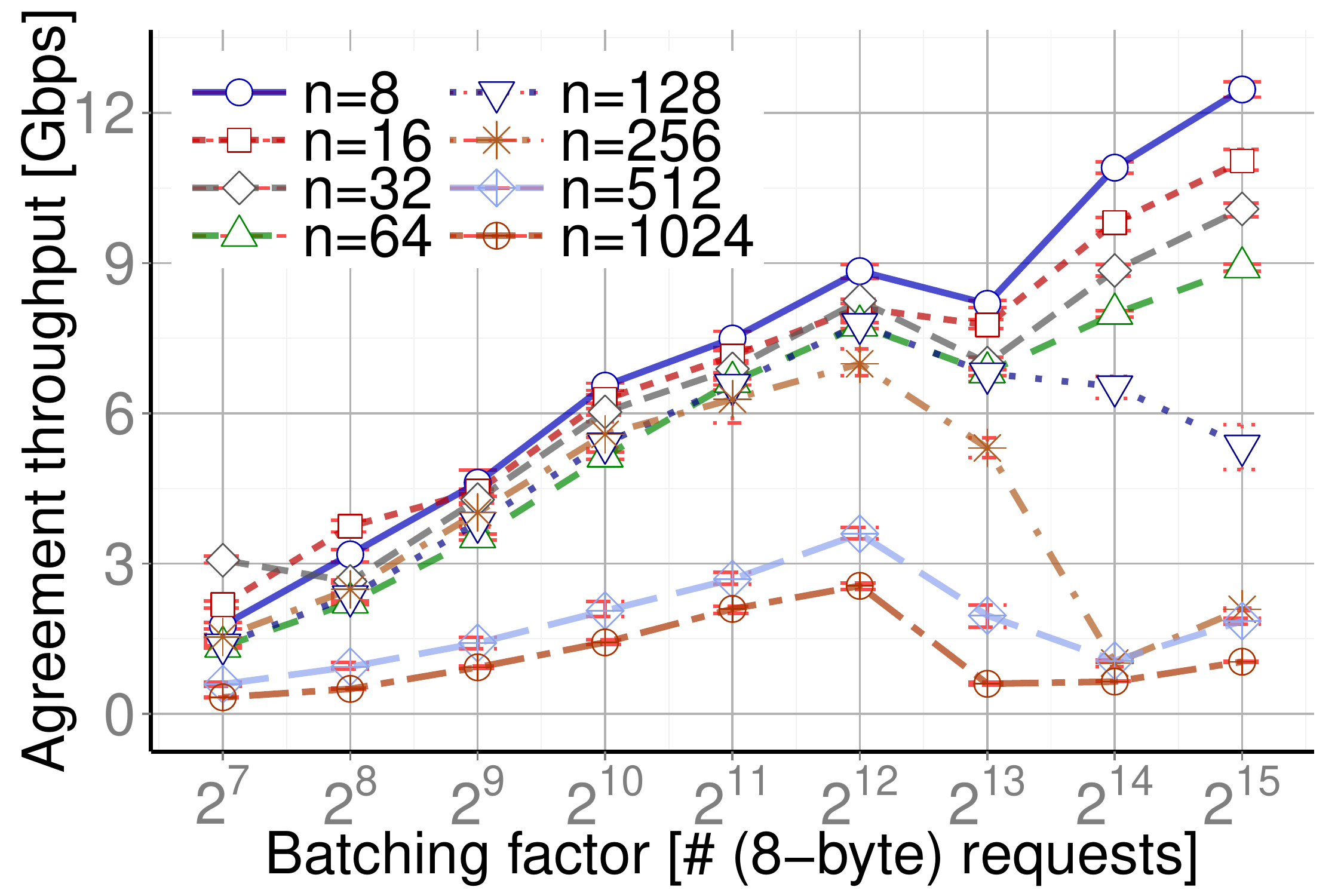}
}
~
\subfloat[\label{fig:allconcur_tcp_throughput} {\LSC{}-TCP [XC40]}] {
\includegraphics[width=.235\textwidth]{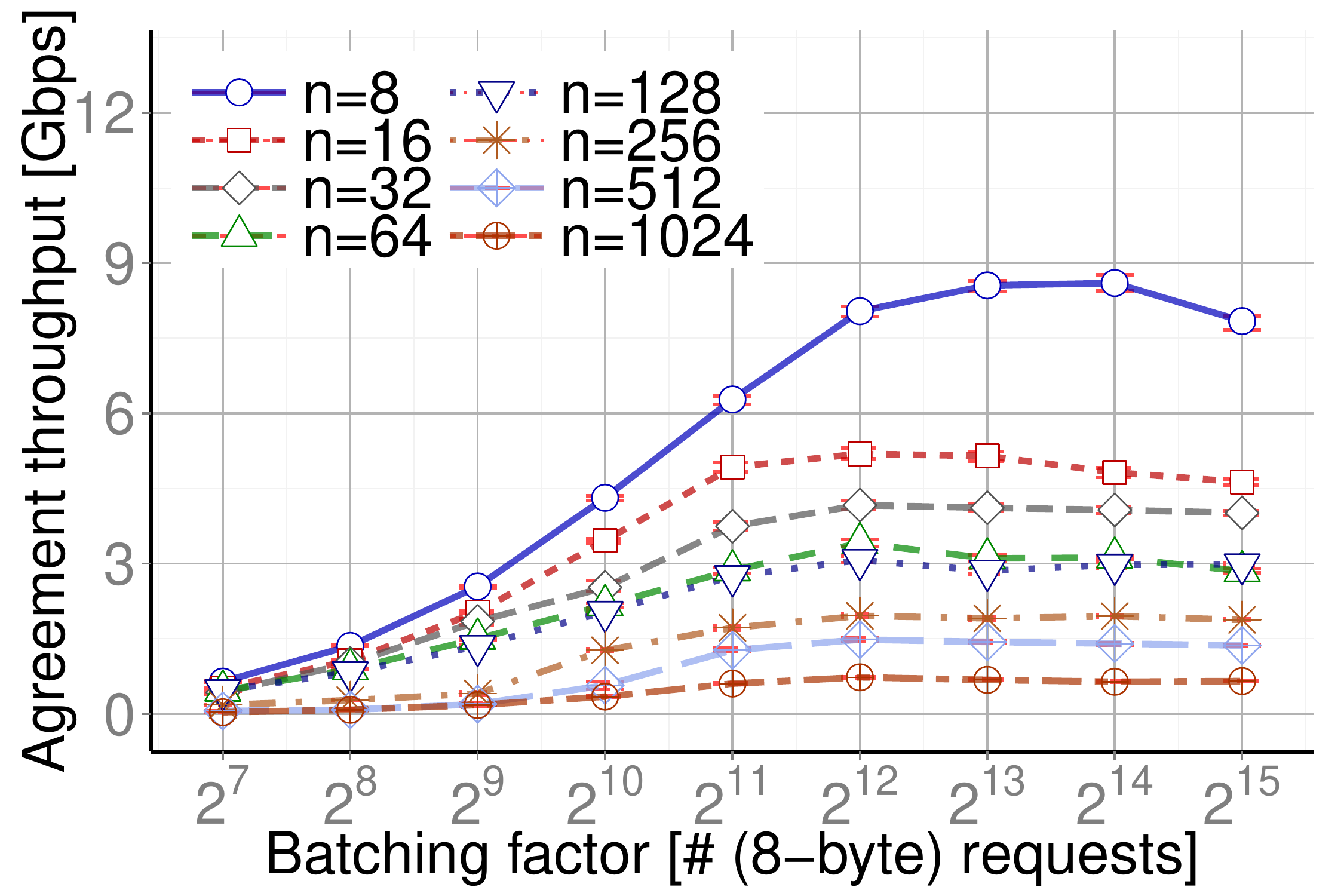}
}
~
\subfloat[\label{fig:libpaxos_throughput} {Libpaxos [TCP / XC40]}]{
\includegraphics[width=.235\textwidth]{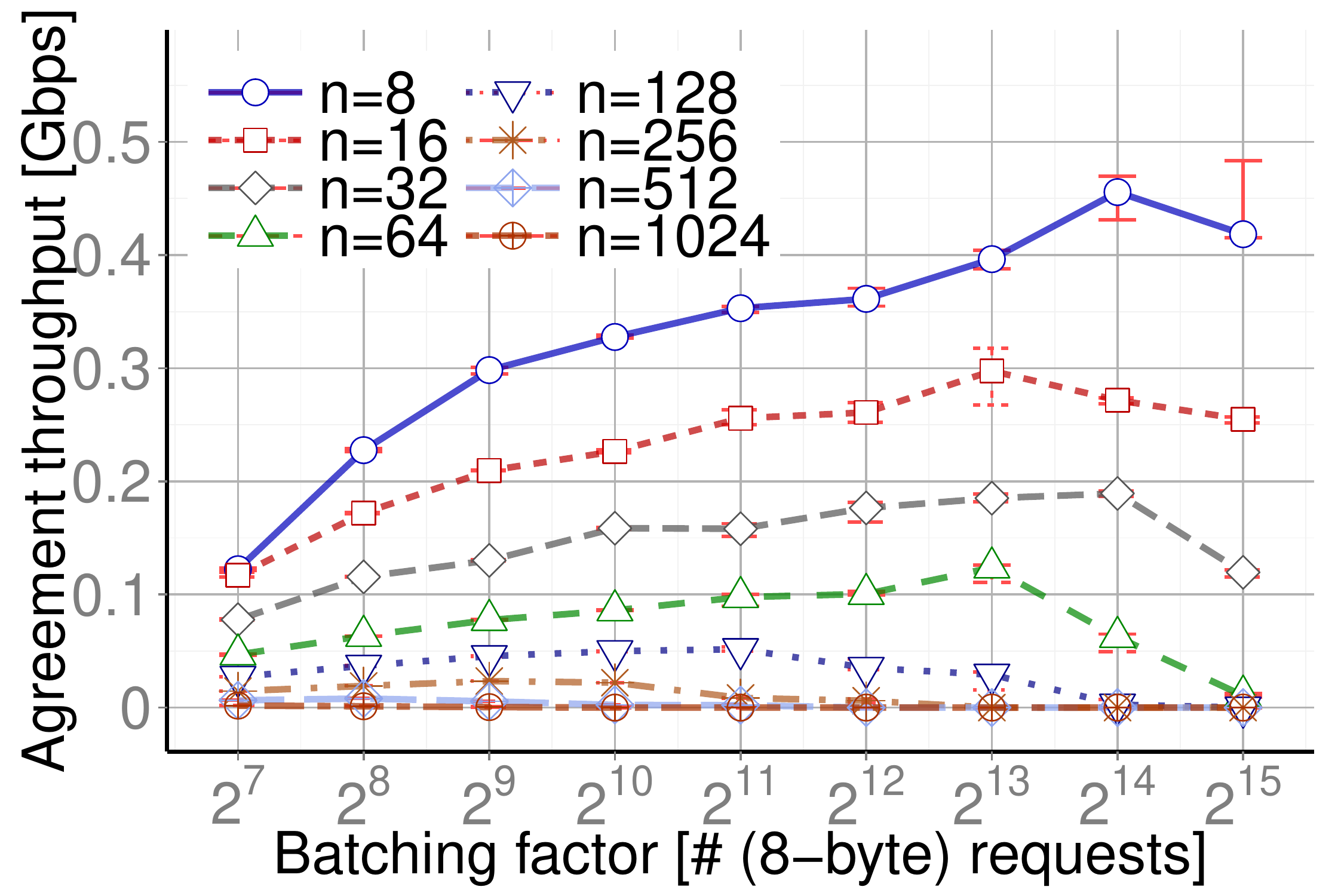}
}
~
\subfloat[\label{fig:allconcur_volume} {\LSC{}-TCP [XC40]}]{
\includegraphics[width=.235\textwidth]{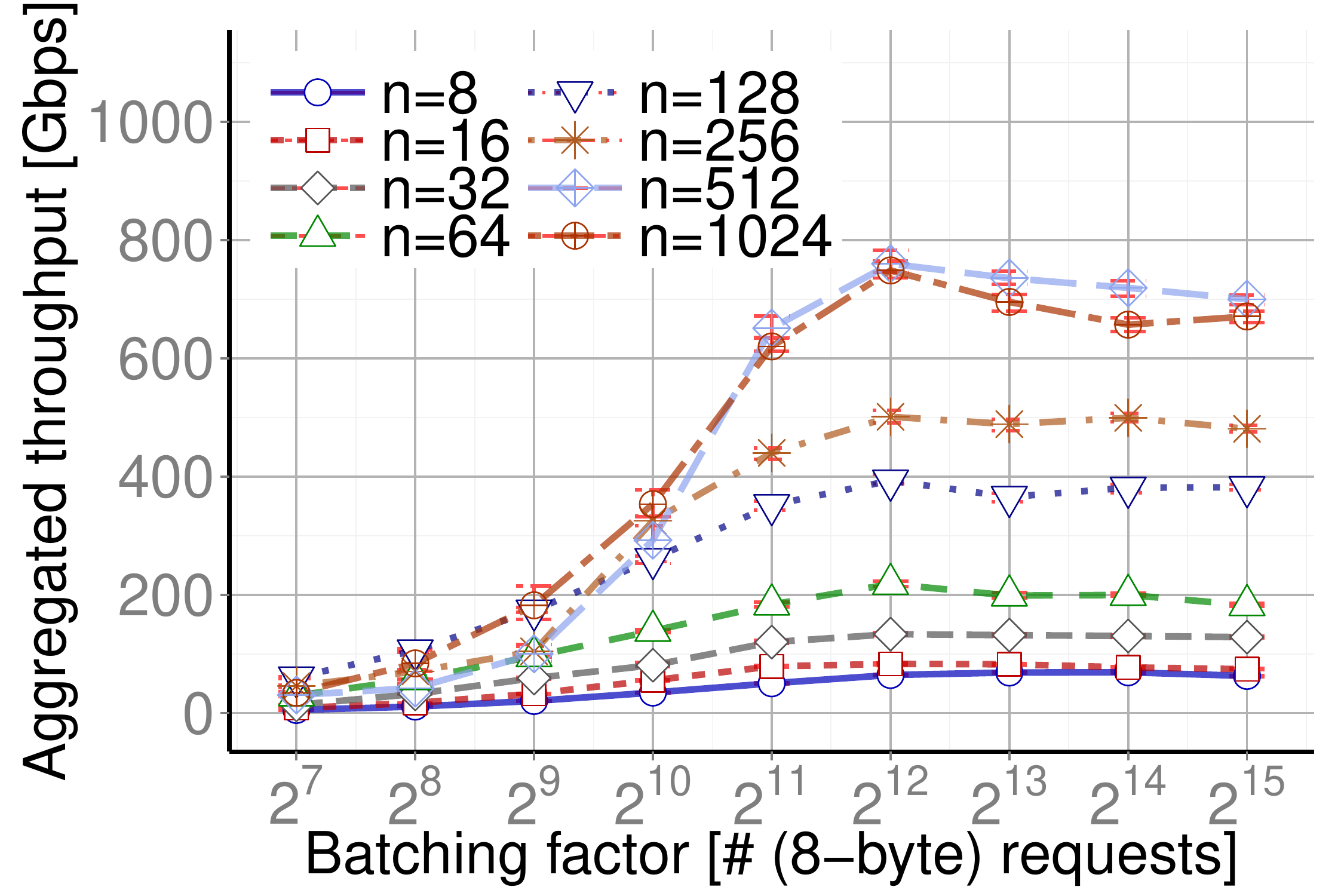}
}
\caption{(a) Unreliable agreement vs. (b) \LSC{} vs. (c) leader-based agreement---batching factor effect on the agreement throughput.
(d) Batching factor effect on the aggregated throughput.}
\label{fig:batching_factor}
\end{figure*}
\fi

\textbf{Travel reservation systems.}
In this scenario, each server's rate of generating requests is bounded by its rate of answering queries. 
We consider a benchmark where 64-byte requests are generated with a constant rate per server~$r$.
Since the \emph{batching factor} (i.e., the amount of requests packed into a message) is not bounded, the system becomes unstable once 
the rate of generating requests exceeds the agreement throughput; this leads to a cycle of larger messages, leading to longer 
times, leading to larger messages etc. A practical deployment would bound the message size and reduce the inflow of requests.
Figures~\ref{fig:allconcur_ibv_rate_srv} and~\ref{fig:allconcur_tcp_rate_srv} plot the agreement latency as a function of $r$; the measurements were obtained on the 
IB-hsw system.
By using \LSC{}-IBV, $8$ servers, each generating 100 million requests per second, reach agreement in $35\mu s$; 
while $64$ servers, each generating 32,000 requests per second, reach agreement in less than $0.75ms$.
\LSC{}-TCP has $\approx3\times$ higher latency.

\textbf{Multiplayer video games.}
In this scenario, the state is updated periodically, e.g., once every $50ms$ in multiplayer video games~\cite{Bharambe:2008:DEL:1402958.1403002,Bharambe:2006:CDA:1267680.1267692}; 
thus, such systems are latency sensitive. 
Moreover, similarly to travel reservation systems, each server's rate of generating requests is bounded;
e.g., in multiplayer video games, each player performs a limited number of actions per minute (APM), 
i.e., usually $200$ APM, although expert players can exceed $400$ APM~\cite{lewis2011corpus}.
To emulate such a scenario, we deploy \LSC{} on the XC40 system; although not designed for video games, 
the system enables large-scale deployments. 
Figure~\ref{fig:allconcur_gaming} plots the agreement latency as a function of the number of players, 
for 200 and 400 APM. Each action causes a state update with a typical size of $40$ bytes~\cite{Bharambe:2008:DEL:1402958.1403002}.
\LSC{}-TCP supports the simultaneous interaction among $512$ players with an agreement latency of $28ms$ for 200 APM 
and $38ms$ for 400 APM. Thus, \LSC{} enables so called epic battles~\cite{epic_battles}.

\ifdefined\TECHREP
\textbf{Distributed exchanges.}
In this scenario, the servers are handling a globally constant rate of requests, e.g., the clients' orders to a distributed exchange service. 
To emulate such a scenario, we deploy \LSC{} on the XC40 system; although this system is not geographically distributed, 
we believe the results are good indicators of \LSC{}'s behavior at large scale. 
Figure~\ref{fig:allconcur_rate_sys} plots the agreement latency as a function of the system's rate of generating request.
An \LSC{}-TCP deployment across 8 servers handles 100 million (40-byte) requests per second with latencies below $90\mu s$;
while, across 512 servers, it can handle one million requests per second with latencies below $20ms$.
The $4\times$ increase in agreement latency for 1,024 servers is due to the  $11\times$ redundancy of the $G_S$ digraph, 
necessary for achieving our reliability target of $6$-nines (see Table~\ref{tab:sim_digraph}). 
\fi

\textbf{\LSC{} vs. unreliable agreement.}
To evaluate the overhead of providing fault-tolerance, we compare \LSC{} to 
an implementation of unreliable agreement. 
In particular, we use MPI\_Allgather~\cite{mpi-3.1} to disseminate all messages to every server.
We consider a benchmark where every server delivers a fixed-size message per round (fixed number of requests). 
Figures~\ref{fig:mpi_allgather_throughput} and~\ref{fig:allconcur_tcp_throughput} plot the agreement throughput 
as a function of the batching factor. The measurements were obtained on the XC40 system;
for a fair comparison, we used Open MPI~\cite{gabriel04:_open_mpi} over TCP to run the benchmark.
\LSC{} provides a reliability target of 6-nines with an average overhead of~$58\%$. Moreover, 
for messages of at least $2,048$ ($8$-byte) requests, the overhead does not exceed~$75\%$.

\textbf{\LSC{} vs. leader-based agreement.}
We conclude \LSC{}'s evaluation by comparing it to Libpaxos~\cite{libpaxos}, 
an open-source implementation of Paxos~\cite{Lamport:1998:PP:279227.279229,Lamport2001} over TCP.
In particular, we use Libpaxos as the leader-based group in the deployment described in Section~\ref{sec:comparison}.
The size of the Paxos group is five, sufficient for our reliability target of $6$-nines.
We consider the same benchmark used to compare to unreliable agreement---each server A-delivers a fixed-size message per round.
Figures~\ref{fig:allconcur_tcp_throughput} and~\ref{fig:libpaxos_throughput} plot the agreement throughput as a function of the batching factor;
the measurements were obtained on the XC40 system.
The throughput peaks at a certain message size, indicating the optimal batching factor to be used.
\LSC{}-TCP reaches an agreement throughput of $8.6Gbps$, equivalent to $\approx135$ million (8-byte) requests per second (see Figures~\ref{fig:allconcur_tcp_throughput}). 
As compared to Libpaxos, \LSC{} achieves at least $17\times$ higher throughput (see Figure~\ref{fig:libpaxos_throughput}).
The drop in throughput (after reaching the peak), for both \LSC{} and Libpaxos, is due to the TCP congestion control mechanism.

\LSC{}'s agreement throughput decreases with increasing the number of servers. The reason for this performance drop is twofold.
First, to maintain the same reliability, more servers entail a higher degree for $G$ 
(see Table~\ref{tab:sim_digraph}), hence, more redundancy.
Second, agreement among more servers entails more synchronization. 
Yet, the number of agreeing servers is an input parameter. 
Thus, a better metric to measure \LSC{}'s actual performance is the aggregated throughput. 
Figure~\ref{fig:allconcur_volume} plots the aggregated throughput corresponding to the agreement throughput from Figures~\ref{fig:allconcur_tcp_throughput}. 
\LSC{}-TCP's aggregated throughput increases with the number of servers and it peaks at $\approx750Gbps$ for 512 and 1,024 servers. 

\section{Related Work}

Many existing algorithms and systems can be used to implement atomic broadcast; 
we discuss here only the most relevant subset. 
D{\'e}fago, Schiper, and Urb\'{a}n provide a general overview of atomic broadcast algorithms~\cite{Defago:2004:TOB:1041680.1041682}. 
They define a classification based on how total order is established: by the sender, 
by a sequencer or by the destinations~\cite{Chandra:1996:UFD:226643.226647}. 
\LSC{} uses destinations agreement to achieve total order, i.e., agreement on a message set.
Yet, unlike other destinations agreement algorithms, 
\LSC{} is entirely decentralized and requires no leader.

Lamport's classic Paxos algorithm~\cite{Lamport:1998:PP:279227.279229,Lamport2001} 
is often used to implement atomic broadcast. Several practical systems
have been proposed~\cite{Boichat2003,Kirsch:2008:PSB:1529974.1529979,Corbett:2012:SGG:2387880.2387905,Marandi2012}.
Also, a series of optimizations were proposed, such as distributing the load among all servers 
or out-of-order processing of not-interfering requests~\cite{moraru2013there,lamport2004generalized,Mao:2008:MBE:1855741.1855767}. 
Yet, the commonly employed simple replication scheme is not designed to scale to hundreds of instances.

State machine replication protocols are similar to Paxos but often claim
to be simpler to understand and implement. Practical implementations
include ZooKeeper~\cite{Hunt:2010:ZWC:1855840.1855851}, Viewstamped Replication~\cite{Liskov2012},
Raft~\cite{Ongaro2014}, Chubby~\cite{Burrows:2006:CLS:1298455.1298487} and DARE~\cite{Poke:2015:DHS:2749246.2749267} among others.
These systems commonly employ a leader-based approach which makes them fundamentally unscalable.
Increasing scalability comes often at the cost of relaxing the consistency model~\cite{li2013zht,DeCandia:2007:DAH:1323293.1294281}.
Moreover, even when scalable strong consistency is provided~\cite{Glendenning:2011:SCS:2043556.2043559}, these systems aim to increase data reliability, 
an objective conceptually different than distributed agreement.

Bitcoin~\cite{bitcoin} offers an alternative solution to the (Byzantine fault-tolerant) atomic broadcast problem: It uses \emph{proof-of-work} to order the 
transactions on a distributed ledger. In a nutshell, a server must solve a cryptographic puzzle in order to add a block of transactions to 
the ledger. Yet, Bitcoin does not guarantee \emph{consensus finality}~\cite{vukolic2015quest}---multiple servers solving the puzzle may 
lead to  a fork (conflict), resulting in branches. Forks are eventually solved by adding new blocks. 
Eventually one branch outpaces the others, thereby becoming the ledger all servers agree upon. To avoid frequent forks, Bitcoin controls the 
expected puzzle solution time to 10 minutes and currently limits the block size to 1MB, resulting in limited performance, i.e., around seven transactions per 
second. To increase performance, Bitcoin-NG~\cite{eyal2016bitcoin} uses proof-of-work to elect a leader that can add blocks until a new leader 
is elected. Yet, conflicts are still possible and consensus finality is not ensured. 

\section{Conclusion}

In this paper we present \LSC{}: a distributed agreement system that relies on a novel 
leaderless atomic broadcast algorithm. \LSC{} uses a digraph $G$ as overlay network; 
thus, the fault-tolerance $f$ is given by $G$'s vertex-connectivity $k(G)$ and can be adapted 
freely to the system specific requirements. 
We show that \LSC{} achieves competitive latency and throughput for 
\ifdefined\PAPER
two real-world scenarios. 
\fi
\ifdefined\TECHREP
three real-world scenarios. 
\fi
In comparison to Libpaxos, \LSC{} achieves at least $17\times$ higher throughput for the considered scenario.
We prove \LSC{}'s correctness under two assumptions---$f < k(G)$ and a perfect failure detector. 
Moreover, we show that if $f \geq k(G)$, \LSC{} still guarantees safety, 
and we discuss the changes necessary to maintain safety when relaxing the assumption of a perfect failure detector.

In summary, \LSC{} is highly competitive and, due to its decentralized approach, enables hitherto unattainable system 
designs in a variety of fields.

\ifdefined\TECHREP
\textbf{Acknowledgements.}
\fi
\ifdefined\PAPER
\begin{acks}
\fi
This work was supported by the German Research Foundation (DFG) as part of the Cluster 
of Excellence in Simulation Technology (EXC 310/2) at the University of Stuttgart.
\ifdefined\PAPER
We thank Michael Resch for support; our shepherd Samer Al Kiswany and the anonymous reviewers; 
Nitin H. Vaidya, Jos\'{e} Gracia and Daniel Rubio Bonilla for helpful discussions;
and Holger Berger for providing support with the InfiniBand machine.
\fi
\ifdefined\TECHREP
We thank Michael Resch for support; Nitin H. Vaidya, Jos\'{e} Gracia and Daniel Rubio Bonilla for helpful discussions;
and Holger Berger for providing support with the InfiniBand machine.
\fi
\ifdefined\PAPER
\end{acks}
\fi

{
\ifdefined\TECHREP
\small
\bibliographystyle{abbrv}
\fi
\ifdefined\PAPER
\bibliographystyle{ACM-Reference-Format}
\fi
\bibliography{references}}

\end{document}